\documentclass[11pt,letterpaper,reqno,oneside]{article}
\usepackage{booktabs,amsmath,amsfonts,amssymb,graphicx,bm}
\usepackage[authoryear]{natbib}
\usepackage{relsize}
\usepackage{times}
\usepackage{dsfont}
\usepackage{etoolbox}
\usepackage{float}
\usepackage{epigraph}
\usepackage{ushort}
\usepackage{hyperref}
\usepackage{xcolor}
\usepackage{tikz}
\usetikzlibrary{snakes,arrows,calc, patterns, positioning, shapes.geometric, decorations.pathreplacing,decorations.markings}
\usepackage{relsize}
\usepackage{pgfplots}
\usepackage{setspace}
\usepackage{mathtools} 
\usepackage{bbm}
\usepackage{subfig}
\usepackage{amsthm}
\pgfplotsset{compat=1.10}
\usepgfplotslibrary{fillbetween}
\makeatletter
\renewcommand{\paragraph}{%
	\@startsection{paragraph}{4}%
	{\z@}{1.75ex \@plus 1ex \@minus .2ex}{-0.7em}%
	{\normalfont\normalsize\bfseries}%
}
\makeatother
\hypersetup{
  colorlinks   = true, 
  urlcolor     = blue, 
  linkcolor    = blue, 
  citecolor   = blue 
}

\newtheorem{theorem}{Theorem}
\newtheorem{claim}{Claim}
\newtheorem{corollary}{Corollary}
\newtheorem{lemma}{Lemma}
\newtheorem{proposition}{Proposition}
\newtheorem{obs}{Observation}

\DeclareMathOperator*{\argmax}{arg\,max}

\def \qed {\hfill \vrule height6pt width 6pt depth 0pt}

\usepackage{datetime}
\newdateformat{datestyle}{\monthname[\THEMONTH], \THEYEAR}

\title{\scshape Advisors with Hidden Motives\thanks{I am grateful for guidance from Debraj Ray, Erik Madsen and Ariel Rubinstein, and for long discussions with Joshua Weiss and Samuel Kapon. I also thank Margaret Meyer, David Pearce, Ludvig Sinander, Dilip Abreu, Nageeb Ali, Ricardo Alonso, Arjada Bardhi, Heski Bar-Isaac, Dhruva Bhaskar, Sylvain Chassang, Navin Kartik, Laurent Mathevet, Luis Rayo, Mauricio Ribeiro, Dezs\"{o} Szalay, Nikhil Vellodi and Basil Williams  for their helpful comments.}}

\author{Paula Onuchic \\ University of Oxford}
\date{\datestyle\today}

\begin{document}
\maketitle
\begin{abstract}
I study a model of advisors with hidden motives: a seller discloses information about an object's value to a potential buyer, who doesn't know the object's value or how profitable the object's sale is to the seller (the seller's motives). I characterize optimal disclosure rules, used by the seller to \emph{steer} sales from lower- to higher-profitability objects. I investigate the effects of a mandated transparency policy, which reveals the seller's motives to the buyer. I show that, by removing the seller's \emph{steering incentive}, transparency can dissuade the seller from disclosing information about the object's value, and from acquiring that information in the first place. This result refines our understanding of effective regulation in advice markets, and links it to the commitment protocol in the advisor-advisee relation.
\end{abstract}

\section{Introduction}

People frequently take advice from advisors with hidden motives: broker-dealers and other financial advisors counsel investors, but also receive sales commissions from financial product providers; digital influencers provide product reviews to their followers, but these are often sponsored content; doctors inform patients of the effectiveness of different drugs and procedures, but may be rewarded by pharmaceutical companies; magazines and newspapers selectively publish pieces of reporting that align with their editorial bias. In all the mentioned settings, information receivers understand that information providers may be biased, but are not fully aware of the extent of the conflict. For example, in the context of financial brokers, clients understand that brokers receive sales commissions from some product providers, but may not know the size of the commissions on each product. 

Such conflicts of interest between advisors and advisees receive much regulatory attention. In all these contexts, the most common recommendation, and often instituted regulation, is \emph{transparency}: an advisor should let an advisee know the real interests behind their recommendations.\footnote{Think of the SEC mandating that financial advisors disclose commissions they receive from financial product providers, or Instagram asking digital influencers to mark their posts as ``paid content'' when they are sponsored by brands. In section \ref{sec:apl}, I mention some empirical literature that considers the effects of transparency in the context of financial advice. Ershov and Mitchell (2022) measure the effects of disclosure regulations on Instagram.} The informal intuition behind this ubiquitous regulation is that, if a person knows the intention behind the information they receive, then they can take the advice with ``a grain of salt.'' The advisor, anticipating the advisees ``healthy skepticism,'' should then have the right incentives to provide informative advice.

In this paper, I scrutinize that intuition, and investigate whether an advisee is always better off when their advisor's motives are transparent. To do so, I propose a model where a seller discloses information about an object to a buyer, who can choose to acquire the object. The seller is an advisor with \emph{hidden motives}, because the buyer does not know how interested they are in the object's sale. In this context, I show that a policy that informs the buyer about the seller's interests does not necessarily improve the informativeness of the seller's advice to the buyer. Indeed, in some contexts the transparency policy transparency decreases the seller's willingness to disclose information about the object's value; and it may also hinder the seller's incentives to acquire information about the object's value in the first place. The main results in this paper show that, in a communication protocol \emph{with commitment}, the effectiveness of transparency policies is linked to the curvature of the buyer's demand for the object. This novel observation is in contrast with results about the effectiveness of transparency policies under communication protocols \emph{without commitment}.\footnote{Communication protocols with commitment have been increasingly studied in the information design literature since the important contributions in Kamenica and Gentzkow (2011) and Rayo and Segal (2010); and are in contrast with previous literature which mainly studied cheap talk as in Crawford and Sobel (1982) and disclosure games as in Grosman (1981) and Milgrom (1981). The terms ``with commitment'' versus ``without commitment'' refer to environments in which the sender \emph{commits} to a communication strategy before observing the realization of the state of the world (the value of the object, in this paper's context) or strategically communicates after seeing the state realization.}

I introduce the model in more detail in section \ref{sec:mod}. The seller wishes to sell an object to a buyer who does not know the object's \emph{value} or how \emph{profitable} its sale is to the seller. The seller knows the object's profitability and sees some signal about its value. Prior to observing the object's profitability or the realization of the signal about its value, the seller commits to a rule to disclose signal realizations to the buyer. This rule assigns to each realization of the signal and each profitability a probability that the realization is disclosed to the buyer. 
The buyer is Bayesian and updates their belief about the object's value based on any information the seller reveals and on the seller's disclosure policy itself. The probability that they purchase the object is given by the buyer's ``demand function,'' which is increasing in the posterior expected value of the object.

When choosing what information to share with the buyer, the seller wants to disclose evidence that the object is valuable, so as to incentivize the buyer to purchase it. They also know that they can hide bad signal realizations, but if they do so,  the buyer becomes skeptical and reads non-disclosure itself as a negative sign, which lowers the probability of sale. In Theorem \ref{th:1}, I show that optimal disclosure rules follow a simple \emph{threshold structure}. Each signal realization about the object's value is classified as either ``good news,'' if above some endogenously determined threshold, or ``bad news.'' For each good news realization, there is a profitability threshold such that the realization is disclosed if the object's profitability is above the threshold, and concealed otherwise. Conversely, each piece of bad news is concealed if the object is sufficiently profitable, and revealed otherwise.  Using such a disclosure policy, the seller effectively \emph{steers} purchases from low-  to high-profitability objects.

Theorem \ref{th:1} also characterizes the seller's optimal profitability thresholds, and Corollaries \ref{prop1} and \ref{prop3} show that their shape depends on the curvature of the buyer's demand function. Next, I consider introducing a transparency policy, which institutes that the buyer observes the object's profitability to the seller. I show that, under transparency, the seller cannot use strategic disclosure in order to steer sales across objects with different profitability levels; and so the policy alters their incentives to disclose information about the object's value to the buyer. Interestingly, it is not necessarily true that this policy \emph{increases} the seller's incentives to inform the buyer. Indeed, Corollary \ref{cor:3} and Proposition \ref{pr:p3} show that if the buyer's demand function is strictly concave, then under mandated transparency the seller voluntarily discloses more evidence to the buyer when their motives are hidden. (The opposite conclusion holds when the buyer's demand function is strictly convex.) 

The intuition behind this result is that, when the seller's motives are hidden, some selective information disclosure allows them to steer sales from low-  to high-profitability objects. Importantly, in order to effectively steer sales, the seller \emph{must disclose some information}. Under the transparency policy, precisely because the buyer knows the seller's motives, the ``steering incentive'' for information disclosure vanishes, sometimes leading the seller to disclose strictly less information. Having considered the effects of transparency on the seller's incentives to disclose information, in section \ref{acq}, I augment the model with an ex-ante stage in which the seller invests in acquiring information about the object's value. I find that the ``steering incentive'' is also an important driver of the seller's decision to invest in information acquisition. 

In section \ref{sec:alt}, I consider the robustness of the results mentioned above with respect to changes in the communication protocol. 
One exercise considers a disclosure protocol with no commitment. I first show that, due to the sender's hidden motives, partially uninformative equilibria may not unravel, as in traditional voluntary disclosure games. The unravelling logic fails because, upon observing non-disclosure, the buyer may not know whether the seller has ``good news'' but negative profitability or ``bad news'' but positive profitability.\footnote{This result is in line with Seidmann and Winter's (1997) observation that in disclosure environments where the sender's preferred action depends on their type, there may exist equilibria in which their type is not always revealed in equilibrium.}
But if a transparency policy is introduced as before, which reveals to the buyer the object's profitability, then all partially uninformative equilibria unravel, and the unique equilibrium outcome is full disclosure. Therefore imposing transparency of the seller's motives unambiguously improves their willingness to share information with the buyer.

This last result is in contrast with the effects of transparency policies found under the ``disclosure with commitment'' protocol. This contrast highlights that the informal intuition behind transparency policies --- that if the buyer knows the intention behind the information they receive, then they take the seller's advice with ``the right degree of skepticism,'' and the seller has the right incentives to provide informative advice --- is anchored on an understanding of the advice market as a ``no commitment'' environment. While no commitment is a useful benchmark, there are reasons (such as reputation-building in repeated interactions) why a degree of ``commitment'' exists in the relationship between the advisor and the advisee. See  section \ref{sec:apl}, for example, where I discuss the main assumptions of the model in light of an application where the seller is a  financial advisor and the buyer is an investor. In that light, the current paper shows that if there is some degree of commitment that relationship, then transparency of the advisor's motives is not always a sufficient method to regulate the advice market.


\subsection{Related Literature}

This paper contributes to the literature studying the disclosure of hard evidence, started with Grossman (1981) and Milgrom (1981).\footnote{For a survey, see Milgrom (2008)} In particular, in my model the sender uses a disclosing technology as in Dye (1985) and Jung and Kwon (1988), whereby evidence realization is either entirely disclosed to the receiver or concealed. My model mainly departs from that literature by (i) considering a disclosure problem where the sender has hidden motives, and (ii) considering a protocol where the sender commits to a disclosure rule prior to the realization of the evidence. Regarding the commitment assumption, sections \ref{sec:alt} and \ref{sec:apl} discuss at length the implications of commitment on the communication outcome.

In assuming that the sender can commit ex-ante to a signaling technology, my paper relates to the large literature on information design, mainly stemming from Kamenica and Gentzkow (2011) and Rayo and Segal (2010).\footnote{For a survey, see Kamenica (2019)} Indeed, my model is a version of Rayo and Segal's (2010) model where communication happens through disclosure of hard evidence. Rayo and Segal (2010) do not describe their model as representing advisors \emph{with hidden motives}, and do not consider the effect of mandated transparency on the informativeness of the sender, which is the focus of my analysis.\footnote{The term “transparent motives” was coined by Lipnowski and Ravid (2020) to refer to a problem where the sender’s preferences are state-independent. In my model, under mandated transparency, the seller acts \emph{as if} they have transparent motives, in that sense.} Moreover, I contribute to the understanding of the benchmark model in Rayo and Segal (2010), by characterizing optimal disclosure for general receiver ``demand functions,'' while most of Rayo and Segal's (2010) characterization results apply to the linear specification.\footnote{With nonlinear demand functions, the optimal disclosure rule in my model sometimes ``pools ordered prospects'' -- in Rayo and Segal's (2010) language, ordered prospects are two objects whose values and profitabilities are ordered in the same way. For an example, see the optimal disclosure rule depicted in the right panel of Figure \ref{fig2}. One of Rayo and Segal's (2010) main characterization results is that, in the linear benchmark, optimal signals never pool ordered prospects.} Finally, my model also departs from this benchmark by considering costly information acquisition.

Within information design, my paper contributes to a growing literature studying \emph{constrained information design}. These are problems where the sender is subject to additional constraints, beyond Bayesian plausibility, when choosing a signal structure. For example, Mensch (2021) and Onuchic \textcircled{r} Ray (2022) consider monotonicity constraints. See Doval and Skreta (2021) and references therein for other examples. This paper considers a sender who is constrained to disclosure strategies. The characterization of optimal signals is a largely open problem in the information design literature and this paper provides such a characterization in an environment which is simplified by the assumption that the sender is constrained to a reasonable class of signals.

Finally, previous literature studies environments with endogenous information acquisition and voluntary disclosure, and argues that voluntary disclosure (as opposed to regulated mandated disclosure), as well as disagreement between sender and receiver, can incentivize information acquisition. Matthews and Postlewaite (1985) argue that by imposing that the seller disclose any evidence they may have about the quality of an object, a regulator can remove their incentives to test the object's quality in the first place. Kartik, Lee and Suen (2017)  show that an advisee may prefer to solicit advice from just one biased expert even when others (of equal or opposite bias) are available, because in the presence of more advisors, each individual expert free rides on the information acquired by the other experts.  In Che and Kartik (2009), though sender and receiver share the same preferences, they hold different priors over the distribution of states. Che and Kartik (2009) show that the sender may invest more in acquiring information when the disagreement between their priors is larger.\footnote{Shishkin (2022) and DeMarzo, Kremer, and Skrzypacz (2019) also study information acquisition by the sender in a Dye (1985) framework. Szalay (2005) and Ball and Gao (2021) study information acquisition by a biased agent in a delegation context.}

In contrast with this literature, my paper argues that not only does mandated transparency of the seller's motives affect their incentives to acquire evidence about the object's value, it also affects their decision to \emph{reveal their previously acquired} evidence to the buyer. This novel observation highlights the differences in the ``advice markets'' regulatory environment under communication protocols with commitment versus without commitment.

\section{Environment}\label{sec:mod}

A seller wishes to sell an object to a buyer who does not know their value for the object, or how profitable the object's sale is to the seller. The object's \emph{value} to the buyer, denoted $x\in\mathcal{X}=[x_{min},x_{max}]$, and its \emph{profitability} to the seller, denoted $y\in\mathcal{Y}=[y_{min},y_{max}]$, with $y_{min}\geqslant 0$, are drawn from some joint distribution commonly known by seller and buyer. We assume the joint distribution and its marginals have no mass points. Denote by $F_Y$ the cdf of the marginal profitability distribution.

Before the buyer decides whether to purchase the object, the seller can convey to them some information about the object's value (see below for details on the communication protocol). After observing any conveyed information, the buyer forms a Bayesian posterior belief about the object's value, with some expected value $\hat{x}$. Given the object's expected value $\hat{x}\in\mathcal{X}$, it is purchased by the buyer with probability $p(\hat{x})$, where $p:\mathcal{X}\rightarrow[0,1]$ is some strictly increasing and differentiable ``demand function.''\footnote{The ``demand function'' $p$ is a reduced-form representation of the buyer's behavior, as a function of their belief about the object's value. This approach assumes that (i) the buyer's purchase decision depends only on their belief about the object's expected value, and (ii) the probability of sale is increasing in the buyer's belief. 

The demand $p$ can be microfounded by many different models. For example, suppose the buyer has some outside option of privately known value and chooses between purchasing the object or keeping their outside option. Further, suppose that, from the seller's perspective, the outside option is distributed according to some continuous cdf $p$. The buyer purchases the object whenever their posterior mean $\hat{x}$ is larger than the outside option. From the seller's perspective, that happens with probability $p(\hat{x})$, which is thus a strictly increasing ``demand function.''} If the object is bought, the seller receives a payoff equal to the object's profitability, $y$. Otherwise, the seller's payoff is $0$. 

\vskip.3cm
\noindent\textbf{Information Disclosure Protocol.} The seller sees the object's profitability, $y$, as well as a signal of the value of the object to the buyer. Before observing the productivity and the signal realization, the seller \emph{commits} to a rule to  \emph{disclose} signal realizations to the buyer.

The signal, $\pi:\mathcal{X}\times\mathcal{Y}\rightarrow \Delta \mathcal{M}$, is a map between the object's value and profitability and a distribution of messages. Each signal realization $m\in\mathcal{M}$, if disclosed, induces on the buyer a Bayesian posterior belief about the object's value, whose expectation is the buyer's \emph{posterior mean} value $\hat{x}(m)=\mathbb{E}(x|m,\pi)$. We assume that the signal structure $\pi$ is such that, given a signal realization $m$, the buyer's posterior mean is \emph{independent} of the object's profitability: for all $y\in\mathcal{Y}$, $\mathbb{E}(x|m,\pi,y)=\mathbb{E}(x|m,\pi)=\hat{x}(m)$. 

Given this independence, we refer to a signal realization and to the posterior mean it induces interchangeably. In other words, if the buyer observes a message $m$, they interpret it as ``the object's expected value is $\hat{x}(m)$,'' independently of their belief about the object's profitability. Further, under this independence assumption, we can equate the signal $\pi$ with its induced distribution of posterior means. This distribution, call its cdf $F_{X|y}$, may still vary with the object's profitability $y$ (despite the independence assumption), because the underlying value of the object and its profitability may be correlated.

Each of the realizations can be disclosed by the seller to the buyer, and this choice may depend on the profitability of the object. A \emph{disclosure rule} is a measurable map from profitabilities and signal realizations into disclosure probabilities: $d:\mathcal{Y}\times \mathcal{X} \rightarrow [0,1]$. When a signal realization $\hat{x}$ is disclosed, the receiver's posterior mean after observing it is exactly $\hat{x}$. If otherwise a signal realization is not disclosed, the receiver's posterior mean is computed using Bayes' Rule, accounting for the disclosure rule. Formally, if $\int_\mathcal{Y}\int_\mathcal{X}\left[1-d(x,y)\right]dF_{X|y}(x)dF_Y(y)\neq0$,\footnote{Otherwise, non-disclosure is ``off-path,'' and we fix $x^{ND}$ at some value in $\mathcal{X}$.}

\begin{equation}
x^{ND}=\frac{\int_\mathcal{Y}\int_\mathcal{X}x\left(1-d(x,y)\right)dF_{X|y}(x)dF_Y(y)}{\int_\mathcal{Y}\int_\mathcal{X}\left(1-d(x,y)\right)dF_{X|y}(x)dF_Y(y)},
\label{avg}
\end{equation}
which is the expected value conditional on non-disclosure.  Define $y^{ND}$ analogously to be the average object profitability given that a realization is not disclosed.

\vskip.3cm
\noindent\textbf{Hidden Motives and Informativeness.} In the environment described above, the buyer's decision to purchase the object depends only on the object's \emph{value} -- the buyer does not inherently care about the \emph{profitability} of the sale to the seller. However, the buyer understands that the the object's sale may be more or less profitable, and that the seller may be more or less motivated to share information about the object's value depending on the object's profitability. 

Because the buyer does not observe the seller's profitability, we say that the seller is an advisor \emph{with hidden motives}. This paper tackles the question of whether the buyer would be better off if they were able to observe the advisor's motives. To that end, we at times consider an alternative benchmark \emph{with transparency}, where the buyer observes the signal realizations disclosed by the seller, as well as the object's profitability.

\section{Optimal Disclosure}\label{disclosure}
Suppose the seller commits to a disclosure rule $d$. The probability that the object is sold, conditional on its profitability being $y$ is  
\begin{align}
P(y,d)=\int_\mathcal{X}\left[d(x,y)p(x)+(1-d(x,y))p(x^{ND})\right]dF_{X|y}(x).\label{eq:p1}
\end{align}
To understand (\ref{eq:p1}), first remember that $F_y$ is the distribution of signal realizations given that the object has profitability $y$. Suppose a signal realization $x$ is disclosed, which happens with probability $d(x,y)$. Then the object is sold with probability $p(x)$, which is reflected in the first term inside the integral of (\ref{eq:p1}). As for the second term, with probability $1-d(x,y)$ the realization $x$ is not disclosed. In that case, the object is sold with probability $p(x^{ND})$, where $x^{ND}$ is as given in (\ref{avg}). 

The seller's expected payoff from committing to disclosure rule $d$ is then
\begin{align}
\Pi(d)=\mathbb{E}\left[yP(y,d)\right]=\mathbb{E}(y)\mathbb{E}\left[P(y,d)\right]+\text{Cov}\left[y,P(y,d)\right].\label{eq:p2}
\end{align}

In (\ref{eq:p2}), I split the seller's payoff into two terms, expressing that the seller's objective can be seen as twofold. According to the first term, the seller wishes to maximize the overall expected probability of sale, which is multiplied by the average profitability. Per the second term, they seek to maximize the covariance between the object's profitability and its probability of sale. 

Theorem \ref{th:1} provides a \emph{threshold characterization} of disclosure rules that maximize the seller's value. There is a threshold value $\bar{x}$ such that each signal realization is classified as either \emph{good news}, if $x$ is larger than $\bar{x}$, or \emph{bad news}, if $x<\bar{x}$. Additionally, for good news realization, there is a profitability threshold such that the realization is disclosed if and only if the object's profitability is above that threshold. Conversely, each bad news realization is disclosed if and only if the object's profitability is below some threshold.

\begin{theorem}
\label{th:1}
An optimal disclosure rule exists and every optimal rule $d^*$  has a threshold structure: There is a threshold value $\bar{x}\in\mathcal{X}$ and a threshold profitability $\bar{y}:\mathcal{X}\rightarrow \mathcal{Y}$ such that $d^*$ almost everywhere satisfies $d^*(x,y)\in\{0,1\}$ and
$$d^*(x,y)=1\Leftrightarrow (x-\bar{x})(y-\bar{y}(x))\geqslant 0.$$

\noindent The threshold value satisfies $\bar{x}=x^{ND}$, and, for $x\neq x^{ND}$, $\bar{y}(x)$ satisfies
\begin{align}
\bar{y}(x)=y^{ND}\left[\frac{p'(x^{ND})(x^{ND}-x)}{p(x^{ND})-p(x)}\right].\label{wth}
\end{align}
\end{theorem}
I will now argue that any disclosure rule that does not satisfy the described threshold structure described can be improved upon, but the complete proof of Theorem \ref{th:1} is in the Appendix. 
Start with a disclosure rule $d$ which does not satisfy the threshold structure, and let $x^{ND}$ be its implied non-disclosure posterior mean. Now define an alternative rule, $\hat{d}$, that discloses each realization $x$ with the same probability as $d$, but has a threshold structure. That is, if $x\leqslant x^{ND}$, let
\begin{align}
\label{p1}
\hat{d}(x,y)=\begin{cases}
1\text{, if }y\leqslant\hat{y}(x)\\
0\text{, if }y>\hat{y}(x)
\end{cases}
\end{align}
\label{p2}
and if $x> x^{ND}$, let
\begin{align}
\hat{d}(x,y)=\begin{cases}
0\text{, if }y<\hat{y}(x)\\
1\text{, if }y\geqslant\hat{y}(x)
\end{cases}
\end{align}
where the thresholds $\hat{y}$ are calibrated such that, for each realization $x$,
\begin{align}
\label{eq:ddhat}
\int_\mathcal{Y}\hat{d}(x,y)dF_{Y|x}(y)=\int_\mathcal{Y}d(x,y)dF_{Y|x}(y),
\end{align}
where $F_{Y|x}$ is the profitability distribution conditional on a signal realization $x$. Condition (\ref{eq:ddhat}) implies that $d$ and $\hat{d}$ induce the same $x^{ND}$, and so $\hat{d}$ satisfies all the requirements of the threshold structure in Theorem \ref{th:1}. Figure \ref{fig:th1A} shows a disclosure rule $d$ that does not have a threshold structure and an improvement $\hat{d}$. 

By moving from $d$ to $\hat{d}$, the seller shifts the disclosure probability of bad news to low profitability objects and of good news to high profitability objects, while maintaining the distribution of posterior mean values that is induced on the buyer. It is easy to see (and I argue formally in Appendix \hyperref[appA]{A}), that: 1. $d$ and $\hat{d}$ produce the same overall probability of sale, because the distribution of posterior mean values is unchanged; and 2. $\hat{d}$ induces a strictly larger covariance between sales and profitability than $d$, because the change increases the probability that very profitable objects are sold, and decreases that probability for less profitable objects. These facts imply that $\hat{d}$ yields a strictly larger expected seller payoff than $d$, proving that the optimal rule must have a threshold structure.

\begin{figure}[t!]
\begin{minipage}{.47\textwidth}
\hskip-.3cm\begin{tikzpicture}[scale=.8]
    \begin{axis}[thick,smooth,no markers,
        xmin=0,xmax=1,
        ymin=0,ymax=1,
        x label style={at={(axis description cs:.6,-.15)},anchor=north},
    y label style={at={(axis description cs:-.05,.6)},anchor=south},
        xlabel=Signal Realization $x$,
        ylabel=Profitability $y$,
        xtick={.4},
      xticklabels={$x^{ND}$},
        ytick=\empty,
        yticklabels={}
]
 	\addplot[name path=A, no marks,very thin, black, -] expression[domain=0:.4,samples=200]{1};
	\addplot[name path=B, no marks,very thin, black, -] expression[domain=.4:.8,samples=200]{0};
	\addplot[name path=E, no marks,very thin, black, -] expression[domain=.8:1,samples=200]{1};
	\addplot[name path=C, no marks,very thin, red!50, -] expression[domain=0:.4,samples=200]{.4+2*x*(x-.4)};
	\addplot[name path=D, no marks,very thin, red!30, -] expression[domain=.4:.8,samples=200]{.4-2*(x-.4)*(x-1)};
	\addplot[name path=F, no marks,very thin, red!30, -] expression[domain=.8:1,samples=200]{.6+2*(x-.4)*(x-1)};

        \addplot[red!30] fill between[of=A and C];
        \addplot[red!30] fill between[of=B and D];
        \addplot[red!30] fill between[of=E and F];
    \end{axis}
   \draw[ dashed](2.75,0)--(2.75,5.7);
    \end{tikzpicture}
\end{minipage}
\begin{minipage}{.47\textwidth}
\hskip.5cm\begin{tikzpicture}[scale=.8]
    \begin{axis}[thick,smooth,no markers,
        xmin=0,xmax=1,
        ymin=0,ymax=1,
        x label style={at={(axis description cs:.6,-.15)},anchor=north},
    y label style={at={(axis description cs:-.05,.6)},anchor=south},
        xlabel=Signal Realization $x$,
        ylabel=Profitability $y$,
        xtick={.4},
      xticklabels={$x^{ND}$},
        ytick=\empty,
        yticklabels={}
]
 	\addplot[name path=A, no marks,very thin, black, -] expression[domain=0:.4,samples=200]{1};
	\addplot[name path=B, no marks,very thin, black, -] expression[domain=.4:.8,samples=200]{0};
	\addplot[name path=E, no marks,very thin, black, -] expression[domain=.8:1,samples=200]{0};
	\addplot[name path=C, no marks,very thin, red!50, -] expression[domain=0:.4,samples=200]{.4+2*x*(x-.4)};
	\addplot[name path=D, no marks,very thin, red!30, -] expression[domain=.4:.8,samples=200]{.4-2*(x-.4)*(x-1)};
	\addplot[name path=F, no marks,very thin, red!50, -] expression[domain=.8:1,samples=200]{.4-2*(x-.4)*(x-1)};

        \addplot[red!30] fill between[of=A and C];
        \addplot[red!30] fill between[of=B and D];
        \addplot[red!50] fill between[of=E and F];
    \end{axis}
   \draw[ dashed](2.75,0)--(2.75,5.7);
    \end{tikzpicture}    \end{minipage}
\caption{\small{Each panel depicts a disclosure rule, with the colored areas representing zones of no disclosure and the white areas representing zones of disclosure. The disclosure rule depicted on the right panel has a threshold structure, and is an improvement over the disclosure rule in the left panel, which does not have a threshold structure.}}
\label{fig:th1A}
\end{figure}
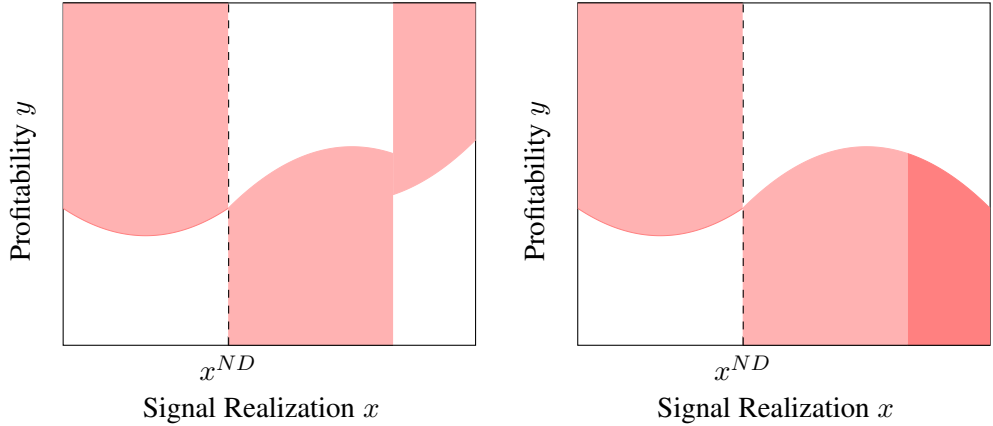

Computationally, Theorem \ref{th:1} achieves a significant simplification of the seller's problem. All disclosure rules with the described threshold structure are almost everywhere defined by two numbers: $x^{ND}$ and $y^{ND}$. As such, the seller's problem can be approached as a two-dimensional maximization problem.
We see from equation (\ref{wth}) that the distribution of profitability and signal realizations -- and the correlation between the object's value and its profitability -- affects the shape of the optimal thresholds only through $x^{ND}$ and $y^{ND}$. Further, some features of the optimal profitability threshold are fully defined by $p$, the buyer's ``demand function'' -- these features are further explored in Propositions (\ref{prop1}) and (\ref{prop3}) below. 

\vskip.3cm 
\noindent\textbf{Linear Demand.} Corollary \ref{prop1} applies Theorem \ref{th:1} when $p$ is an affine function.

\begin{corollary}[to Theorem \ref{th:1}]
If the demand function $p$ is affine, then all optimal disclosure rules have a threshold structure, with $\bar{x}=x^{ND}\in\text{int}\left(\mathcal{X}\right)$ and $\bar{y}(x)=y^{ND}\in\text{int}\left(\mathcal{Y}\right)$ for all $x\in\mathcal{X}$.\footnote{The fact that the thresholds must be interior does not follow directly from Theorem \ref{th:1}. We provide a proof of that claim in the Appendix.}
\label{prop1}
\end{corollary}

\begin{figure}[t!]
\begin{minipage}{.45\textwidth}
\hskip-.5cm\begin{tikzpicture}[scale=.7]
    \begin{axis}[thick,smooth,no markers,
        xmin=0,xmax=1,
        ymin=0,ymax=1,
        xtick={.4},
      xticklabels={$\bar{x}$},
        ytick={.3},
        yticklabels={$\bar{y}$},
        title style={at={(.8,-.4)},anchor=north,yshift=-1}
]
 	\addplot[name path=A, no marks,very thin, black, -] expression[domain=0:.4,samples=200]{1};
	\addplot[name path=B, no marks,very thin, black, -] expression[domain=.4:1,samples=200]{0};
	\addplot[name path=C, no marks,very thin, black, -] expression[domain=0:.4,samples=200]{.3};
	\addplot[name path=D, no marks,very thin, black, -] expression[domain=.4:1,samples=200]{.3};

        \addplot[gray] fill between[of=A and C];
        \addplot[gray] fill between[of=B and D];
    \end{axis}
    \draw[very thick](2.75,0)--(2.75,5.7);
    \draw[very thick](0,1.7)--(6.85,1.7);
    \end{tikzpicture}
\end{minipage}
\begin{minipage}{.45\textwidth}
\hskip-.8cm\begin{tikzpicture}[scale=.7]
    \begin{axis}[thick,smooth,no markers,
        xmin=0,xmax=1,
        ymin=0,ymax=1,
        xtick={.4},
      xticklabels={$\bar{x}$},
        ytick={.1,.3,.7},
        yticklabels={Low Profitability, $\bar{y}$, High Profitability},
        title style={at={(.8,-.4)},anchor=north,yshift=-1},
]
 	\addplot[name path=A, no marks,very thin, black, -] expression[domain=0:.4,samples=200]{1};
	\addplot[name path=B, no marks,very thin, black, -] expression[domain=.4:1,samples=200]{0};
	\addplot[name path=C, no marks,very thin, black, -] expression[domain=0:.4,samples=200]{.3};
	\addplot[name path=D, no marks,very thin, black, -] expression[domain=.4:1,samples=200]{.3};

        \addplot[gray] fill between[of=A and C];
        \addplot[gray] fill between[of=B and D];
    \end{axis}
    \draw[very thick](2.75,0)--(2.75,5.7);
    \draw[very thick](0,1.7)--(6.85,1.7);
    \draw[very thick, red](0,4)--(2.75,4);
    \node at (1.3,4.4) {$d^*=0$};
    \draw[very thick, blue](2.75,4)--(6.85,4);
    \node at (4.8,4.4) {$d^*=1$};
        \draw[very thick, blue](0,.54)--(2.75,.54);
    \node at (1.3,.94) {$d^*=1$};
    \draw[very thick, red](2.75,.54)--(6.85,.54);
    \node at (4.8,.94) {$d^*=0$};
    \end{tikzpicture}
    \end{minipage}
\caption{\small{Optimal disclosure rule when $p$ is \textbf{affine}. The gray areas represent signal realizations that are optimally concealed by the seller, and white areas are optimally disclosed. }}
\label{fig1}
\end{figure}
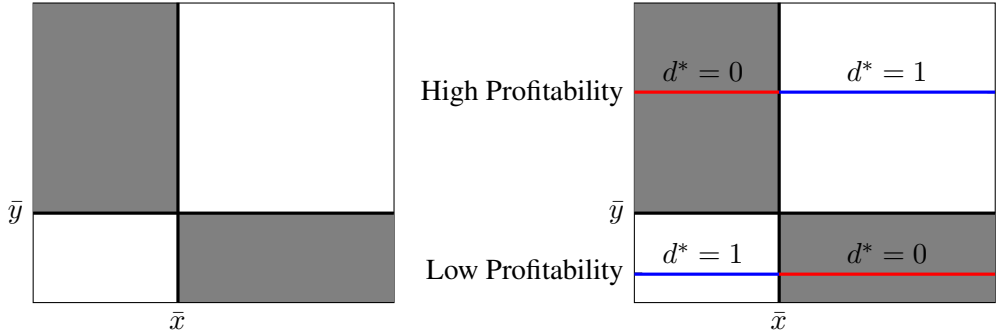

When the demand function $p$ faced by the seller is affine, as an implication of the martingale property of beliefs, all disclosure functions yield the same overall probability of sale. That is, for any two disclosure functions $d$ and $d'$, $\mathbb{E}\left[P(y,d)\right]=\mathbb{E}\left[P(y,d')\right]$. 
This means that, in choosing a disclosure policy, there is no scope for the seller to increase or decrease the overall probability that the buyer purchases the object. Rather, the seller solely distributes this constant sale probability between high and low profitability objects. 

In order to optimally \emph{steer sales}  from low- to high-profitability objects, the seller uses a disclosure rule that assigns objects to two classes, based on their profitability: high profitability, with $y>y^{ND}$, and low profitability, with $y<y^{ND}$. One such disclosure rule is depicted in Figure \ref{fig1}. Given this rule, the buyer's posterior mean for low profitability objects is \emph{at best} equal to $x^{ND}$, which is the posterior mean induced when the signal realizations above $x^{ND}$ are not disclosed. Conversely, the buyer's induced posterior mean for high profitability is \emph{at least} equal to $x^{ND}$, which is induced when signal realizations below $x^{ND}$ are not disclosed.

\vskip.3cm 
\noindent\textbf{Nonlinear Demand.} 
If the demand function $p$ is not affine, then the \emph{amount of information} disclosed by the seller  impacts the overall probability of sale. Observation \ref{obs0} below states that increasing the amount of information about the object's value that is disclosed to the buyer \emph{increases} the overall probability that the object is sold if the demand function is convex, and decreases it if the demand function is concave. To formally state this result, we say that a disclosure rule $d$ has \emph{more disclosure} than a disclosure rule $d'$ if $d(x,y)\geqslant d'(x,y)$ for all $x\in\mathcal{X}$ and $y\in\mathcal{Y}$; and strictly so if this inequality is strict for a subset of $\mathcal{X}\times\mathcal{Y}$ with positive measure.
\begin{obs} 
Suppose disclosure rule $d$ has (strictly) more disclosure than disclosure rule $d'$. Then (i) if $p$ is strictly convex, $d$ yields a (strictly) higher probability of sale than $d'$; and (ii) if $p$ is strictly concave, $d$ yields a (strictly) lower probability of sale than $d'$.
\label{obs0}
\end{obs}
With a nonlinear demand function, the seller sales from low- to high-profitability objects at the expense of the total probability that the object is sold. If the seller faces a convex demand function, he wishes to disclose more information in order to maximize the sale probability. Conversely, if facing a concave demand, overall sale probability increases when information is concealed from the buyer. 

\begin{corollary}[to Theorem \ref{th:1}]If $p$ is strictly convex (concave), any optimal disclosure rule has a threshold structure with a strictly decreasing (increasing) profitability threshold function $\bar{y}(x)$, satisfying $\bar{y}(x^{ND})=y^{ND}$. 
\label{prop3}
\end{corollary}

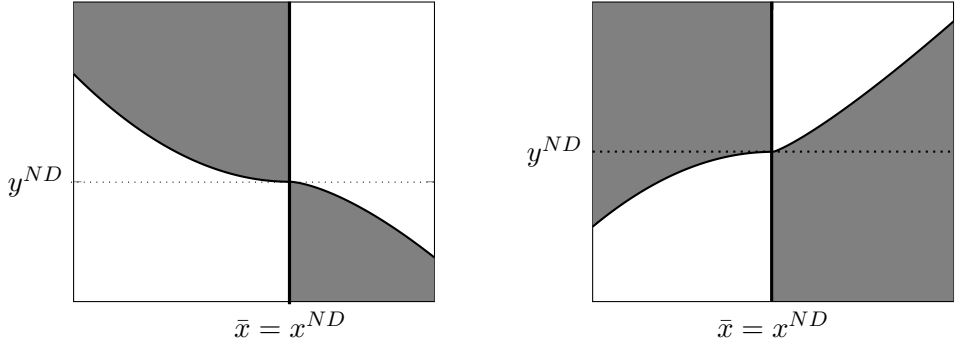
\begin{figure}[t!]
\begin{minipage}{.50\textwidth}
\hskip-.5cm\begin{tikzpicture}[scale=.7]
    \begin{axis}[thick,smooth,no markers,
        xmin=0,xmax=1,
        ymin=0,ymax=1,
        xtick={.6},
        xticklabels={$\bar{x}=x^{ND}$},
        ytick={.4},
        yticklabels={$y^{ND}$},
        title style={at={(.8,-.4)},anchor=north,yshift=-1},
]
 	\addplot[name path=A2, no marks,very thin, black, -] expression[domain=0:.6,samples=200]{1};
	\addplot[name path=B2, no marks,very thin, black, -] expression[domain=.6:1,samples=200]{0};
	\addplot[name path=C2, no marks,thick, black, -] expression[domain=0:.6,samples=200]{(x-.6)^2+.4};
	\addplot[name path=D2, no marks,thick, black, -] expression[domain=.6:1,samples=200]{-(x-.6)^(1.5)+.4};

        \addplot[gray] fill between[of=A2 and C2];
        \addplot[gray] fill between[of=B2 and D2];
    \end{axis}
    \draw[very thick](4.1,-.05)--(4.1,5.7);
    \draw[dotted, thin](-.05,2.27)--(6.85,2.27);
    \end{tikzpicture}
\end{minipage}
\begin{minipage}{.50\textwidth}
\begin{tikzpicture}[scale=.7]
    \begin{axis}[thick,smooth,no markers,
        xmin=0,xmax=1,
        ymin=0,ymax=1,
        xtick={.5},
      xticklabels={$\bar{x}=x^{ND}$},
        ytick={.5},
        yticklabels={$y^{ND}$},
        title style={at={(.8,-.4)},anchor=north,yshift=-1}
]
 	\addplot[name path=A3, no marks,very thin, black, -] expression[domain=.5:1,samples=200]{0};
	\addplot[name path=B3, no marks,very thin, black, -] expression[domain=0:.5,samples=200]{1};
	\addplot[name path=C3, no marks,thick, black, -] expression[domain=.5:1,samples=200]{(x-.5)^1.2+.5};
	\addplot[name path=D3, no marks,thick, black, -] expression[domain=0:.5,samples=200]{-(x-.5)^2+.5};

        \addplot[gray] fill between[of=A3 and C3];
        \addplot[gray] fill between[of=B3 and D3];
    \end{axis}
    \draw[very thick](3.4,0)--(3.4,5.7);
    \draw[dotted, thick](0,2.85)--(6.85,2.85);
    \end{tikzpicture}
    \end{minipage}
\caption{\small{Optimal disclosure rule when $p$ is \textbf{strictly convex} (left panel) and \textbf{strictly concave} (right panel). The gray areas represent signal realizations that are optimally concealed by the seller, and the white areas are optimally disclosed.}}
\label{fig2}
\end{figure}

An optimal disclosure rule when $p$ is strictly convex is represented in the left panel of Figure \ref{fig2}. 
A striking difference when comparing to the linear demand case is that some ``very bad news'' may be optimally disclosed even when the object is highly profitable; and, conversely, some ``very good news'' are disclosed to the receiver even when the object has low profitability. This feature is derived from the tension between the desire to maximize the probability of sale (which pushes towards disclosure) and to covary the probability of sales with profitability (which requires some concealment). 
In the opposite case, where $p$ is strictly concave, due to an analogous tension, the sender optimally conceals some ``very bad news'' about low-profitability objects and some ``very good news'' about high-profitability objects. This is depicted in the right panel of  Figure \ref{fig2}.

Corollary \ref{prop3} shows that the ``shape'' of optimal disclosure rules depends on the curvature of the demand function $p$. In the current model, the demand function is taken as a primitive, and its curvature does not necessarily reflect some economically interpretable property. However, in a related model, Hwang, Kim and Boleslavsky (2019) study an oligopolistic market where agents compare the expected value of objects offered by different firms, and purchase the highest-value object. In that model, the ``demand function'' for a particular firm's object is given by the distribution of the maximum value among the objects offered by that firm's competitors. Hwang, Kim and Boleslavsky (2019) show that, as the number of competitors increases, that demand function becomes ``more convex.'' From that perspective, the convexity of the demand function is related to the competitiveness of the market.

\subsection{Is Transparency Good?}

Let's consider the following policy intervention, which makes the motives of the seller transparent to the buyer: When the buyer observes the seller's ``advice'' -- either the disclosed signal realization or non-disclosure -- they also perfectly observe how profitable the object's sale is to the seller. This should be thought of as a transparency mandate imposed on the seller by some regulator, forcing them to always reveal their interest in the object's sale.

Naturally, anticipating that the buyer will see the profitability of the object, the seller optimally uses disclosure strategy which differs from the payoff maximizing disclosure rule without transparency. Indeed, with transparency, the seller's problem becomes continuum of \emph{separate} disclosure problems, one for each profitability. 
\begin{lemma}
With mandated transparency, the seller chooses for each $y\in\mathcal{Y}$ a disclosure policy $d(\cdot,y):\mathcal{X}\rightarrow [0,1]$ to maximize
$P(y,d(\cdot,y))$, where
$$P(y,d(\cdot,y))=\int_\mathcal{X}\left[d(x,y)p(x)+(1-d(x,y))p(x_y^{ND})\right]dF_{X|y}(x)$$ 
\begin{equation}
\label{eqq1}\text{and }x_y^{ND}=\frac{\int_\mathcal{X}[1-d(x,y)] xdF_{X|y}(x)}{\int_\mathcal{X}[1-d(x,y)] dF_{X|y}(x)}.
\end{equation}
\label{lem:1}
\end{lemma}
Under transparency, the seller cannot use strategic disclosure in order to steer sales across objects with different profitability levels. Rather, they separately maximize the probability that the object with each profitability level gets sold. Proposition \ref{propp1} below describes the seller's optimal disclosure strategy under transparency. 

\begin{proposition}
Consider a mandated transparency environmentand  hidden motives environment. Let $d$ and $d'$ be optimal disclosure policies to the seller in these respective environments. 

\begin{enumerate}
\item If $p(\cdot)$ is strictly \textbf{convex}, mandated transparency improves evidence disclosure, that is, $d$ has \textbf{more disclosure} than $d'$.

\item If $p(\cdot)$ is strictly \textbf{concave}, mandated transparency harms evidence disclosure, that is, $d$ has \textbf{less disclosure} than $d'$.

\end{enumerate}
\label{propp1}
\end{proposition}

Proposition \ref{propp1} delineates how the mandated transparency policy affects de seller's willingness to voluntarily disclose evidence about the object's value to the buyer. In short, if the demand function is convex, then the mandated transparency policy increases the set of evidence pieces that the seller would voluntarily choose to disclose to the buyer. In fact, under that demand regime, any optimal disclosure policy under mandated transparency involves the voluntary disclosure of \emph{all evidence} -- so that the optimal $d$ involves full disclosure, regardless of the object's profitability. Conversely, if the demand function is concave, then the mandated transparency has the opposite effect on optimal evidence disclosure, and the optimal disclosure policy is for the seller to not disclose any of the realized evidence.

To evaluate whether transparency is good or bad in terms of the informativeness of the seller's advice to the buyer, it is not enough to consider the evidence that is voluntarily disclosed by the seller. Rather, we must also consider the information that is ``involuntarily'' conveyed to the buyer directly when they observe the object's profitability. Remember that objects with different profitabilities may have different value (and evidence) distributions, and therefore some information is conveyed directly through the mandated transparency of the seller's motives. To evaluate the overall impact of the policy on the informativeness of the seller's advice to the buyer, we consider its impact on the distribution of the buyer's beliefs induced by the policy and the seller's optimal disclosure policy. Given a disclosure rule $d$, let $F^B(\cdot,d)$ be the distribution of posterior means observed by the buyer.

Denote the transparency policy by $\tau\in\{0,1\}$, with $\tau=1$ indicating the mandated transparency environment, and $\tau=0$ the hidden motives environment. Given a disclosure rule $d$ and transparency policy $\tau$, let $F^B(\cdot,d,\tau)$ be the distribution of posterior means observed by the buyer.\footnote{For a given $d$ and $\tau$, $F^B(x,d,\tau)$ is equal to\begin{align*}\int_\mathcal{Y}\int_{[x_{min},x)}d(\hat{x},y)dF_{X|y}(\hat{x})dF_Y(y)+\int_\mathcal{Y}\int_\mathcal{X}(1-d(\hat{x},y))\mathbbm{1}\big\{x_y^{ND}\leqslant x\big\}dF_{X|y}(\hat{x})dF_Y(y),\end{align*}
where for every $y\in\mathcal{Y}$, $x_y^{ND}=x^{ND}$ as given in (\ref{avg}) if $\tau=0$; and $x_y^{ND}$ is as defined  in (\ref{eqq1}) if $\tau=1$.}
We say the buyer is \emph{better informed} under the pair $(d,\tau)$ than under the pair $(d',\tau')$ if $F^B(\cdot,d,\tau)$ is more informative than $F^B(\cdot,d',\tau')$ in the Blackwell order. Equivalently, if $F^B(\cdot,d,\tau)$ is a mean preserving spread of $F^B(\cdot,d',\tau')$. By evaluating the policy in terms of its implied Blackwell informativeness to the buyer, I am agnostic as to what the buyer's objective is, and therefore on what the is ``surplus'' that the policy-maker wishes to maximize. As we know from previous literature, the buyer's value under any decision problem (for which the mean posterior is a sufficient statistic) is higher under the more Blackwell informative policy. 

Corollary \ref{cor:3} below follows from Proposition \ref{propp1} and states that the mandated transparency policy is beneficial to the buyer whenever the demand function is convex; but may be harmful when the demand function is concave. This latter conclusion holds whenever no information about the expected value of the object is directly conveyed by the disclosure of the object's profitability, that is, if  $\mathbb{E}[x|y] = \mathbb{E}[x]$ for every $y\in\mathcal{Y}$. 

\begin{corollary}[to Proposition \ref{propp1}]
Consider a mandated transparency environment and  hidden motives environment.

\begin{enumerate}
\item If $p(\cdot)$ is strictly \textbf{convex}, the buyer is better informed under mandated transparency.

\item If $p(\cdot)$ is strictly \textbf{concave}, and $\mathbb{E}[x|y] = \mathbb{E}[x]$ for every $y\in\mathcal{Y}$, then the buyer is  better informed in the hidden motives environment.
\end{enumerate}
\label{cor:3}
\end{corollary}

Proposition \ref{pr:p3} below compares  optimal disclosure rules when the seller's motives are hidden versus made transparent with some probability. The timing is as follows: First, the seller commits to a disclosure rule, understanding that the buyer will observe the object's profitability with some probability $\tau\in[0,1]$. Next, the buyer observes the disclosed signal and, with probability $\tau$, also the object's profitability. Finally, the buyer makes their purchase decision. 

When mandated transparency is probabilistic, then optimal disclosure rules with/without transparency are not always ranked in terms of the disclosure order used in Proposition \ref{propp1}, as that order is incomplete (requiring increased disclosure for every possible evidence realization). However, Proposition \ref{pr:p3} shows that, when the buyer's non-disclosure posteriors (about the object's value and profitability) are the same with and without transparency, then the same ranking as in Corollary \ref{cor:3} holds. That is, transparency improves disclosure if the demand function $p$ is convex, and hinders it if $p$ is concave.
\begin{proposition}
Consider two mandated transparency environment, in which the buyer sees the object's profitability with probability $\tau$ and $\tau'$ respectively, with $1\geqslant\tau\geqslant\tau'\geqslant0$.  Let $d$ and $d'$ be optimal disclosure policies to the seller in the respective environments. 

\begin{enumerate}
\item If $p(\cdot)$ is strictly \textbf{convex}, mandated transparency does not harm evidence disclosure, that is, $d$ has \textbf{no less disclosure} than $d'$. 

If moreover $\tau'=0$ and $x^{ND}(d)=x^{ND}(d')$ and $y^{ND}(d)=y^{ND}(d')$, then $d$ has \textbf{more disclosure} than $d'$.

\item If $p(\cdot)$ is strictly \textbf{concave}, mandated transparency does not improve evidence disclosure, that is, $d$ has \textbf{no more disclosure} than $d'$. 

If moreover $\tau'=0$ and $x^{ND}(d)=x^{ND}(d')$ and $y^{ND}(d)=y^{ND}(d')$, then $d$ has \textbf{less disclosure} than $d'$.

\end{enumerate}
\label{pr:p3}
\end{proposition}

If the demand function is affine, then under mandated transparency, the seller is indifferent between all disclosure rules. In particular, full transparency is one optimal policy. This result is, again, an implication of the martingale property of beliefs. Interestingly, if there is even a small probability that the seller's motives will not be transparent to the buyer, then the seller optimally discloses \emph{as if} his motives are fully hidden (with the disclosure rule described in Corollary \ref{prop1}). This failure of ``continuity'' in the set of seller's optimal disclosure policies is registered in Proposition \ref{propp2} below.

\begin{proposition}
If the demand function $p$ is affine, and the buyer observes the seller's motives with probability $\tau\in[0,1]$, then

\begin{enumerate}
\item If $\tau=1$, the seller is indifferent between all disclosure rules.

\item If $\tau\neq1$, the set of optimal disclosure rules is invariant to $\tau$ and every 
optimal disclosure rule satisfies the conditions in Corollary \ref{prop1}.
\end{enumerate}
\label{propp2}
\end{proposition}

A direct reading of the results in this section is that they delineate conditions under which a policy maker should or should not institute a transparency policy. But beyond this normative implication, these results also highlight that in information design models, the `alignment between sender and receiver preferences' and the `opaqueness of the sender's motives' are distinct objects; and it is not necessarily true that regulations that reduce the latter would also make the sender's interests more aligned with those of the receiver. And even in contexts where it does improve the alignment between the advisor and advisee's interests --- such as when the demand function is convex, in the model --- it does so because, as a byproduct of the increased transparency of their motives, the sender's effective objective function becomes ``more convex,'' therefore inducing them to optimally disclose more evidence. This is in contrast with the intuition in  a disclosure environment without commitment --- discussed in Section \ref{sec:alt} --- in which the transparency of the seller's motives induces the unravelling of uninformative equilibria.

The comparative statics results above --- with respect to the transparency $\tau$ of the seller's motives --- relate to a recent paper by Curello and Sinander (2022), which studies a general information design environment and establishes a relation between increasing the convexity of the sender's objective and increasing the informativeness of the sender's optimal signal. Although I leave a formalization for future work, intuitively when the demand function $p$ is strictly convex, the regulator can ``increase the convexity'' of the seller's effective objective by increasing the transparency of their motives (and therefore reduce their steering ability). This observation hints at the possibility that effective regulation of advice markets (in communication environments ``with commitment'') should directly target increasing the convexity of advisor's objectives, a goal which is only occasionally achieved through transparency policies.

\section{Hidden Motives and Information Acquisition}\label{acq} 
In the baseline model, the signal about the object's value that is observed by the seller is taken as given. I now study how the seller's hidden motives affect their incentive to acquire information about the object's value in the first place, knowing that they can selectively disclose the acquired information to the buyer. To study this problem, we focus on the linear demand case.
\vskip.3cm
\textbf{[Assumption]} For the remainder of the section, we assume that $p$ is affine.

\subsection{Signal Acquisition Technology}

Prior to the object's value and profitability realizing, the seller commits to two actions: they \emph{acquire a costly signal} of the object's value and (as before) choose a rule to disclose that signal. The signal acquisition is \emph{overt}, so that the buyer knows the signal acquired by the seller, even if the signal realization itself is not disclosed.

Signals are indexed by their precision $\theta\in\mathbb{R}_+$, and are acquired at a cost $c(\theta)$, where $c$ is strictly increasing. The set of signals available to the seller is such that 
\vskip.3cm
\noindent \textbf{(i)} For every $\theta$ and every $y\in\mathcal{Y}$, $F_{X|y}$ has full support over some interval $\mathcal{X}(\theta)$.\footnote{Along with assumption (ii), this implies that $\theta'>\theta$ $\Rightarrow$ $\mathcal{X}(\theta)\subseteq\mathcal{X}(\theta')$.}
\vskip.3cm
\noindent \textbf{(ii)} More precise signals are more informative than less precise ones: $\theta'> \theta$ implies that, for every $y\in\mathcal{Y}$, $F_{X|y}(\cdot,\theta')$ is a mean-preserving spread of $F_{X|y}(\cdot,\theta)$. 
\vskip.3cm
\noindent \textbf{(iii)} The signal with precision $\theta=0$ is perfectly uninformative: for every $y\in\mathcal{Y}$, $F_{X|y}(\cdot,0)$ is a degenerate distribution.

\subsection{Transparency and Information Acquisition}
We want to consider how mandated transparency of the seller's motives affect their decision to acquire information about the object's value. Consider the policy that discloses the object's profitability to the receiver with probability $\tau\in[0,1]$.

\begin{proposition}
If $\tau=1$, the seller optimally acquires the least precise signal, that is, $\theta=0$.
Moreover, take $\tau>\tau'$ and let $\theta$ and $\theta'$ be signal precisions optimally acquired under transparency policies $\tau$ and $\tau'$, respectively. Then $\theta\leqslant\theta'$.
\label{propn}
\end{proposition}
Proposition \ref{propn} shows that transparency negatively affects the seller's incentives to acquire information about the object's value. Indeed, the seller always acquires a \emph{less precise signal} when faced with a policy that mandates \emph{more transparency}. Moreover, under full transparency, the seller has no incentives to acquire any information, and chooses the least precise (and least costly) signal available.

Remember that, when the seller's motives are hidden, they can strategically disclose information about the object's value to the buyer in order to steer their purchases from low- to high-profitability objects. But in order to selectively inform the buyer, the seller must have information about the object's value in the first place. Indeed, the more precise a signal the seller acquires, the more profitably they can steer the buyer's action. Now consider the policy that mandates transparency: when the seller's motives are (with some probability) made transparent to the buyer, their ability to steer sales between objects with different profitabilities is hindered. In turn, their incentives to acquire information about the object's value are also dampened, and they choose to invest in a less precise signal.

While Proposition \ref{propn} unequivocally establishes that transparency hinders information acquisition, it does not imply that the buyer is unequivocally better informed when the seller's motives are less transparent. Because the seller does not fully disclose the information they acquire, the comparison between the seller's informativeness under two different transparency policies is not always unambiguous (in the Blackwell sense). We evaluate transparency policies in terms of their induced informativeness, as before. The disclosure environment induced by transparency policy $\tau$ is $\delta=(\tau,\theta,d)$ where $\theta$ and $d$ are optimal evidence acquisition and disclosure policies to the seller, given transparency $\tau$. We denote by  $F^B(\cdot,\delta)$ the distribution of posterior means (about the object's value) that is induced on the buyer by the environment $\delta$. 

Proposition \ref{propnn} establishes some comparisons of the informativeness under different transparency regimes. We do so under the assumption that the object's profitability and the distribution of signals about its value are independently distributed. By doing so, we ensure that no information  about the object's value  is directly conveyed to the buyer through the revelation of the object's profitability. Proposition \ref{propnn} then first shows that transparency does not make the seller more Blackwell informative. (Note that, because the informativeness order is not complete, a not more informative signal is not necessarily less informative.) Next, it states that, under full transparency, the seller is Blackwell less informative than under partial or no transparency. 

Finally, I establish some conditions on the structure of signals available to the seller such that more transparency necessarily implies that the seller becomes Blackwell less informative. I say a signal of precision $\theta$ is \emph{regular} if there is a unique pair $(\hat{x},\hat{y})\in\text{int}\left(\mathcal{X}\times\mathcal{Y}\right)$ such that
$$\mathbb{E}_\theta\left[x|(x-\hat{x})(y-\hat{y})<0\right]=\hat{x}\text{ and }\mathbb{E}_\theta\left[y|(x-\hat{x})(y-\hat{y})<0\right]=\hat{y}.$$
Regularity guarantees that the seller has a unique optimal disclosure rule, given acquired precision $\theta$. Next, I say a signal of precision $\theta$ is \emph{symmetric} if, the object's profitability and the distribution of signals about its value are independently distributed for every $\theta$, and $F_X(\cdot;\theta)$ and $F_Y$ are both symmetric.

\begin{proposition}
\label{propnn}
Suppose  $F_{X|y}(\cdot;\theta)$ is independent of $y$. Let $\tau>\tau'$, and let $\delta=(\tau,\theta,d)$ and $\delta'=(\tau'\theta',d')$ be the disclosure environments induced by transparency policies $\tau$ and $\tau'$, respectively. Then:
\begin{enumerate}
\item $F^B(\cdot,\delta)$ is not strictly Blackwell more informative than $F^B(\cdot,\delta')$.
\item If $\tau=1$, then $F^B(\cdot,\delta')$ is Blackwell more informative than $F^B(\cdot,\delta)$.
\item If signals of precision $\theta$ and $\theta'$ are regular and symmetric, then $F^B(\cdot,\delta')$ is Blackwell more informative than $F^B(\cdot,\delta)$.
\end{enumerate}
\end{proposition}
We can also evaluate the seller's informativeness according to some completion of the Blackwell order, call it $\succsim$.\footnote{A completion of the Blackwell order $\succsim$ is a complete order over $\Delta\mathcal{X}$ such that $G,G'\in\Delta\mathcal{X}$ and $G$ is Blackwell more informative than $G'$ implies that $G\succsim G'$.} Say that the set of signals available to the seller is such that \emph{precision is beneficial} if endowing the seller with a more precise signal weakly improves the seller's optimally disclosed signal, according to the order $\succsim$. 

Formally, fix some transparency level $\tau\in[0,1)$, and suppose that, for every $\theta$, the signal of precision $\theta$ is regular -- so that for each $\theta$, there is a unique optimal disclosure rule $d_\theta$. Also note (per Proposition \ref{propp2}) that such optimal disclosure rule is independent of $\tau$, so long as $\tau\neq1$. We say that \emph{precision is beneficial} if $\theta\geqslant\theta'$ implies that $F^B(\cdot,\tau,\theta,d_\theta)\succsim F^B(\cdot,\tau,\theta',d_{\theta'})$ -- or, equivalently, if $F^B(\cdot,\tau,\theta,d_\theta)$ is $\succsim$-more-informative than $F^B(\cdot,\theta',d_{\theta'})$.

One possible such completion of the Blackwell order is the buyer's surplus. In that case, precision is beneficial if a buyer would prefer to be advised by a seller who has more precise information. Corollary \ref{prop5} below shows that, when precision is beneficial, transparency is always detrimental if evaluated in terms of $\succsim$-informativeness.

\begin{corollary}[to Proposition \ref{propn}]
\label{prop5}
Suppose  $F_{X|y}(\cdot;\theta)$ is independent of $y$, and assume that precision is beneficial. Let $\tau>\tau'$, and let $\delta=(\tau,\theta,d)$ and $\delta'=(\tau'\theta',d')$ be the disclosure environments induced by transparency policies $\tau$ and $\tau'$, respectively. Then $F^B(\cdot,\delta')$ is $\succsim$-more-informative than $F^B(\cdot,\delta)$.
\end{corollary}

\subsection{An Example} \label{sec:ex}
Suppose the object's value and profitability are independent and both are distributed according to $U[0,1]$.
Further,  the seller can acquire signals of precision $\theta\in[0,1]$ about the object's value, subject to a strictly increasing cost $c(\theta)$. 

Given the acquired signal, a piece of hard evidence about the object's value realizes, denoted $e\in[0,1]$. With probability $\theta$, $e$ is equal to the true value of the object, and with probability $(1-\theta)$, $e$ is a random number drawn from $U[0,1]$, uncorrelated with the object's true value. The mean posterior implied by a piece of evidence $e\in[0,1]$ is
$$x=\theta e+(1-\theta)\frac{1}{2}.$$
And so the distribution of mean posteriors implied by $\theta$ is $F_{X|y}(\cdot,\theta)=U\left[\frac{1-\theta}{2},\frac{1+\theta}{2}\right]$, for every $y$. It is easy to see that every signal $\theta$ is symmetric, and that $\theta'> \theta$ implies that, for every $y\in\mathcal{Y}$, $F_{X|y}(\cdot,\theta')$ is a mean-preserving spread of $F_{X|y}(\cdot,\theta)$. I also show in the Appendix that every signal $\theta$ is regular.

Because every signal $\theta$ is regular and symmetric, any optimal disclosure rule has a threshold structure with thresholds $\bar{x}=\mathbb{E}(x)=1/2$ and $\bar{y}=\mathbb{E}(y)=1/2$.\footnote{In the proof of Statement 3 in Proposition \ref{propnn}, I show that regularity and symmetry implies that $\bar{x}=\mathbb{E}(x)$ and $\bar{y}=\mathbb{E}(y)$ in any optimal disclosure rule.} In short, if the object has ``above average'' profitability, the seller discloses evidence about that object's value if and only if the evidence induces an ``above average'' posterior. Conversely, if the object has ``below average'' profitability, the seller discloses evidence about that object's value if and only if the evidence induces a ``below average'' posterior. 

Accounting for the acquired precision $\theta$ and the optimal disclosure rule, the distribution of posterior means induced on the buyer is
$$F^B(x,\theta,d)=\begin{cases}
0\text{, if }x<\frac{1-\theta}{2},\\
\frac{1}{2\theta}\left[x-\frac{1-\theta}{2}\right]\text{, if }x\in\left[\frac{1-\theta}{2},\frac{1}{2}\right),\\
\frac{1}{2}+\frac{1}{2\theta}\left[x-\frac{1-\theta}{2}\right]\text{, if }x\geqslant \frac{1}{2}.
\end{cases}$$
In words, with probability $1/2$ the seller does not disclose the signal realization, and the buyer's posterior mean is equal to $x^{ND}=1/2$. With half probability the seller discloses the signal realization -- either below average realizations for less profitable objects or above average realizations for the more profitable objects -- in which case the distribution of posteriors induced on the buyer is $U\left[\frac{1-\theta}{2},\frac{1+\theta}{2}\right]$. 

Because $\theta'>\theta$ implies that $U\left[\frac{1-\theta'}{2},\frac{1+\theta'}{2}\right]$ is a strict mean-preserving spread of $U\left[\frac{1-\theta}{2},\frac{1+\theta}{2}\right]$, we can conclude that, when transparency increases, the seller acquires strictly less precise signals and provides strictly less information to the buyer. (This example illustrates result 3 in Proposition \ref{propnn}.)

\section{Alternative Communication Protocols}\label{sec:alt}
The model makes two main assumptions regarding the communication protocol. The first one is that the seller can commit to a rule to disclose signal realizations, prior to the object's profitability or the ``evidence'' about its value being drawn. This commitment assumption is typical to the information design literature, recently surveyed in Kamenica (2019). Second, we assume that the seller can use only disclosure policies, \`{a} la Dye (1985), rather than committing to more general signal structures. 

In this section, we consider two variations of the model, dropping each of these two assumptions. Under each alternative communication protocol, we investigate whether mandating transparency about the seller's motives improves the informativeness of their advice to the buyer.
\subsection{Disclosure with No Commitment}
In this section, we allow the support of the object's profitability to include negative profitability values, so we may have $y_{min}<0$.\footnote{Note that the characterization of optimal disclosure rules under the commitment protocol, given in Theorem \ref{th:1}, also applies when $y_{min}<0$. I impose the positive profitability assumption in the earlier sections because it simplifies the analysis of the impact of transparency policies.} When $y_{min}<0<y_{max}$, the buyer is unsure whether the profitability is such that the seller wishes to maximize the object's probability of sale or to minimize it. Now consider the following disclosure protocol with no commitment: First, the seller observes the object's profitability and a signal realization. After that, the seller chooses whether to disclose the signal realization (a piece of \emph{evidence} about the object's value) to the buyer. The buyer observes the disclosed evidence -- or the fact that no evidence was disclosed -- and forms a posterior mean about the object's value, taking into account the seller's equilibrium strategy. As before, the buyer does not observe the object's profitability.\footnote{We assume that the seller cannot choose to disclose the object's profitability to the buyer. This assumption is not restrictive, because as we shall see, in equilibrium the seller would never want to disclose the object's profitability.} The equilibrium notion is Perfect Bayesian Equilibrium.\footnote{Because of the evidence structure and our restriction to the evidence disclosure protocol, there is no ``informed principal'' problem when defining an equilibrium.}

\begin{proposition}
For any demand function $p$, an equilibrium exists, and any equilibrium disclosure strategy $d^*$ has a threshold structure:  $d^*(x,y)\in\{0,1\}$ and\footnote{Assuming that, when indifferent, the seller discloses the signal realization.}
$$d^*(x,y)=1\Leftrightarrow (x-\bar{x})y\geqslant 0,$$
for some $\bar{x}\in\mathcal{X}$, satisfying $\bar{x}=\mathbb{E}\left[x|(x-\bar{x})y<0\right]$ if $\{(x-\bar{x})y<0\}\neq \varnothing$.
\label{pr:5}
\end{proposition}
In evidence disclosure models where the advisor's motives are transparent -- for example, in the seminal works of Grossman (1981), Milgrom (1981) and Milgrom and Roberts (1986) -- equilibria feature full information revelation. The main insight is that any candidate equilibrium with some evidence concealment would ``unravel'' due to the receiver's (the buyer's) rational skepticism. Indeed, the threshold equilibria described in Proposition \ref{pr:5} can be consistent with the unravelling logic. Suppose the object's profitability is always positive ($y_{min}\geqslant0$), and conjecture a threshold equilibrium with some evidence concealment -- that is, suppose $\bar{x}\in(x_{min},x_{max})$. Then, because $y$ is always greater than $0$, it must be that $\bar{x}>\mathbb{E}\left[x|(x-\bar{x})y<0\right]$, and thus the equilibrium condition in Proposition \ref{pr:5} is not satisfied. Such a conjectured equilibrium would unravel: the buyer's posterior upon observing non-disclosure would be $\mathbb{E}\left[x|(x-\bar{x})y<0\right]$, which is strictly smaller than $\bar{x}$. Consequently, when the seller draws a signal realization just under $\bar{x}$, they strictly prefer to reveal it to the buyer, which is inconsistent with the initially conjectured equilibrium. An analogous unravelling argument applies if the object's profitability is always negative ($y_{max}\leqslant0$).

However, when profitability can be both positive or negative, so that $0\in(y_{min},y_{max})$, then any equilibrium involves partial disclosure. There is some interior threshold $\bar{x}$ such that the seller discloses (only and) all  ``good news'' about the object ($x\geqslant\bar{x}$) when profitability is positive, and (only and) all ``bad news'' ($x\leqslant\bar{x}$) about the object when profitability is negative. Partially uninformative equilibria do not unravel, because, upon observing non-disclosure, the buyer does not know whether the seller has ``good news'' but negative profitability or ``bad news'' but positive profitability.

In such an equilibrium, the seller's strategy is akin to \emph{sanitization strategies} as in Dye (1985): conditional on any particular profitability level $y$, the seller discloses only sufficiently favorable information to the buyer. (With the caveat that, to a seller with negative profitability, only ``bad news'' are sufficiently favorable.) The existence of partly-uninformative equilibria as described in Proposition \ref{pr:5}  is in line with Seidmann and Winter's (1997) observation that in disclosure environments where the sender's preferred action depends on their type, there may exist equilibria in which their type is not always revealed in equilibrium. Note also that equilibrium disclosure strategies as described in the proposition are independent of the shape of the demand function $p$ -- in contrast to the model with commitment, where optimal disclosure rules differ based on the curvature of $p$. 

Finally, I turn to the question of whether mandated transparency incentivizes the seller to disclose information about the object's value to the buyer. When the seller's evidence structure is taken as given (that is, when we do not consider the seller's problem of acquiring information in the first place), then mandating transparency is unequivocally beneficial for informativeness. Indeed, under transparency, the seller is fully informative.

\begin{proposition}
Under transparency, full disclosure is the unique equilibrium.
\label{pr:6}
\end{proposition}


\subsection{Unconstrained Signaling Technology}
The usual assumption in information design models is that a sender (the seller, in our case) commits a signal a map from states of the world (the pieces of evidence about the object's value) into distributions of messages to inform a receiver (the buyer) about the state. The sender's choice of such a map is unrestricted. Contrastingly, in this paper, I assumed that the sender is restricted to a class of ``signaling strategies'': the class of disclosure rules, whereby the sender's message either fully conveys the information in a piece of evidence, or is ``silent.'' Such silence is a message in itself, which conveys to the receiver that the state of the world is ``one of the evidence-profitability pairs that would lead the sender to stay silent.''

In the \emph{constrained information design} problem I study, Theorem \ref{th:1} provides quite a complete characterization of the seller's optimal messaging strategy. In comparison, a characterization of the sender's optimal signal in the equivalent unconstrained design problem is elusive -- Rayo and Segal (2010) provide a partial characterization of the optimal signal when the demand function $p$ is affine.

Regardless, some of the results in this paper also hold in the ``unconstrained design version'' of the problem. For example, Proposition \ref{propp1}, and its Corollary \ref{cor:3} would still be true, so that transparency may be detrimental to the seller's incentives to relay information about the object's value to the buyer. Further, Proposition \ref{propn} would also be true in that unconstrained benchmark, highlighting that also in that case transparency hinders the sellers motivation to acquire information about the object's value. However, precisely because we cannot characterize optimal signals in the unconstrained problem, it is hard to pinpoint if and when the extra information would be relayed to the receiver. Consequently, the effect of transparency on the seller's overall informativeness to the receiver are unclear.

\section{Discussion of an Application}\label{sec:apl}

The model describes a situation where a seller informs a potential buyer about the value of an object for sale, where the buyer does not know how profitable the sale of such object would be to the seller. 
Below, I interpret features of the model in the context of a brokerage company who provides financial advice to investors, but also receives commissions for the sale of financial products. In the context of financial brokers, clients understand that brokers receive sales commissions from some product sales, but may not know the size of the commissions on each product (and hence their motives are hidden).

There are many other contexts in which people consume information disclosed by entities with hidden motives: 
doctors inform patients of the effectiveness of different drugs and procedures, but are often sponsored by pharmaceutical companies; online platforms gather reviews from product users, but may disclose them coarsely to benefit some producers; magazines and newspapers selectively publish pieces of reporting that align with their editorial bias. In all the mentioned settings, information receivers understand the information providers may be biased, but are not fully aware of the extent of the conflict.

\subsection{Financial Advisors with Hidden Motives}

When an investor opens an account at a brokerage company, they usually get some access to product reports put together by the company's research team. In addition to publicly available data about companies and industries, a report on a particular product includes forecasts and evaluations produced by the team. Upon seeing the provided research, the investor compares the perceived value of the product to an outside option, which could depend on their current appetite for investment, desire to reallocate their current portfolio, or even independent information she may have sourced about the financial product at hand. 

\vskip.3cm
\textbf{Hidden Motives.} The incentives of the advisor/broker may not be fully understood by the investor. For instance, some of the products available in the brokerage system are proprietary products that are  issued or managed by the broker themselves. Upon selling one of these assets, the broker receives extra compensation. Another source of conflict is that third parties commonly pay the broker for marketing and selling their products, which may make the sale of some products more desirable than others. These considerations can fuel the broker's desire to produce advice that steers investors to the more profitable products.  Indeed, it is well documented that brokers conceal bad news about companies in which they have financial interests and that financial advisors' recommendations favor products with high commissions.\footnote{See, for example, Anagol, Cole and Sarkar (2017) and Eckardt and Rathke-Doppner (2010) about insurance brokers; Chalmers and Reuter (2020) on retirement plans; Inderst and Ottaviani (2012.2) on general financial advice. For a survey on quality disclosure and certification, see Dranove and Jin (2010).} 
In line with our modeling assumption, in this context, clients understand that some products are more profitable to the brokers than others, but the specific profitability of each product's sale to the broker is \emph{hidden}.

\vskip.3cm
\textbf{Disclosure.} In the model, the seller produces information about a product and chooses whether to  share the outcome of the research with the buyer or not. This feature of the model matches the existence of regulatory ``Chinese walls'' within financial institutions. To ensure that the analyses produced by a broker are not biased by the companies' profit-seeking interests, various regulations institute walls between research departments and client-facing teams. These walls ensure that the research production is not tainted by the profitability considerations of the brokerage company. As in the model, the reports produced can be understood as unbiased evidence about a product's quality, while the disclosure of this evidence remains under the discretion of the sales team.

\vskip.3cm
\textbf{Regulation.} Because investors often rely on professional advice when making investment decisions, financial advisory is a highly regulated field. (And has become increasingly so since the Great Recession and the institution of the Dodd-Frank Act.) In the United States, the regulation varies across states, and also across categories within the field -- for example, financial advisors and insurance agents are often subject to more stringent regulation than broker-dealers, who often also provide advice. Specific details notwithstanding, much of the regulation focuses on making the interests of financial advisors transparent to investors who consult them. For example, they might be required to disclose other sources of compensation received beyond service fees, professional affiliations with another broker-dealer or securities issuer, and other potential or existing conflicts of interest. Other types of regulation also forbid advisors from receiving some types of compensation. For example, the UK and the Netherlands have altogether imposed bans on commission payments for some types of financial advisors.

\vskip.3cm
\textbf{Commitment.} In information design models, a common interpretation of the commitment assumption is that it stands in for the sender's desire to build and sustain \emph{credibility} with the receiver in a game of repeated interaction. Think of a circumstance where a financial advisor repeatedly interacts with an investor, informing them about  assets' values over a long time. And suppose that this investor follows the advisor's recommendation, but also sees the outcome of their investment decisions (so they can tell, to some extent, whether the advice was sound). Over time, the investor can ``punish'' unsatisfactory advice from the seller with future incredulity -- maybe by not following their recommendation, maybe by seeking another advisor. Through that mechanism, in the repeated game, the advisor's optimal disclosure strategy would be such that the buyer's interpretation of a ``no disclosure'' message coincides with the interpretation yielded in the one-shot game under the commitment assumption. This point is formally made by Best and Quigley (2020), and a similar argument linking commitment and reputation is provided by Mathevet, Pearce and Stacchetti (2022).

A recent paper, Lin and Liu (2022) propose a new notion of credibility for information design problems: A disclosure policy is credible if the sender cannot profit from tampering with her messages while keeping the message distribution unchanged. It's useful to note that all optimal disclosure rules characterized in Theorem \ref{th:1} are credible, according to this definition. In other words, if the sender can commit to a distribution of posterior means to induce on the receiver, then the optimal way to correlate these posterior means with the object's profitability is attained with threshold-structure disclosure rules as described in the Theorem. Consequently, the seller would not ``tamper'' with messages, conditional on maintaining the distribution of induced posteriors.\footnote{More formally, for each $x\in \mathcal X$, a deviation by the sender is ``undetectable'' in Lin and Liu's (2022) sense, only if it keeps the marginal probability of sending that message unchanged. Therefore, any disclosure rule other than the optimal rule as given in Theorem \ref{th:1}, but which  induces that same marginal distribution over messages, must invariably swap the disclosure of a realization $(x,y)$ and a different realization $(x,y')$; i.e., two states with the same evidence realization $x$. But this would never be a profitable deviation for the seller because the optimal disclosure rule already has a monotone structure where the buyer purchases with a weakly higher probability at higher values of profitability $y$. I thank Nageeb Ali for the discussion in this footnote.}

\vskip.3cm
\textbf{Related Literature.}
In a series of papers in 2012, Inderst and Ottaviani (2012.1, 2012.2, 2012.3) propose models of brokers and financial advisors  compensated through commissions.
In their model, competing sellers offer commissions to the advisor, knowing that he will steer business to the seller that offers highest compensation. The price of the asset is set equal to buyers' expected value for it, so that the consumer retains no surplus. An implication is that information is not valuable to the consumer. The authors show that biased commissioned agents may improve welfare by providing  buyers with less information and steering business to high-commission firms who are also more cost efficient. My model takes an alternative approach, taking the sender's distribution of profitability as given and focusing on the value of the information provided to buyers. In my environment, information is always beneficial to the consumer and I show that hidden motives can improve surplus precisely by increasing the amount of information provided to the consumer.

There is also a large literature that studies the provision of information by Credit Rating Agencies that are financed by fees paid by issuers of financial products. Some key papers in this literature are Bolton, Freixas, and Shapiro (2012), Opp, Opp and Harris (2013), Bar-Isaac and Shapiro (2012) and Skreta and Veldkamp (2009). In this literature, the CRA receives payments equally from all issuers of financial products. In my model, the main concern is that the advisor benefits some products over others because they have different profitabilities.

\section*{References}
\noindent Anagol, S., S. Cole, and S. Sarkar. (2017) ``Understanding the Advice of Commissions Motivated Agents: Evidence from the Indian Life Insurance Market.'' \emph{Review of Economics and Statistics}, \textbf{99}, 1-15.

\vskip.3cm\noindent Ball, I. and X. Gao. (2021) ``Benefitting from Bias.'' \emph{working paper}.

\vskip.3cm\noindent Bar-Isaac, H., and J. Shapiro. (2013) ``Ratings Quality over the Business Cycle.'' \emph{Journal of Financial Economics}, \textbf{108}, 62-78.

\vskip.3cm\noindent Best, J., and D. Quigley. (2020) ``Persuasion for the Long Run.'' \emph{working paper}.


\vskip.3cm\noindent Bolton, P., X. Freixas, and J. Shapiro. (2012) ``The Credit Ratings Game.'' \emph{The Journal of Finance}, \textbf{67}, 85-111.
NBR 6023	

\vskip.3cm\noindent Chalmers, J. and J. Reuter. (2020) ``Is Conflicted Investment Advice Better than no Advice?.'' \emph{Journal of Financial Economics}, forthcoming.


\vskip.3cm\noindent Che, Y.K., and N. Kartik (2009) ``Opinions as Incentives.'' \emph{Journal of Political Economy}, \textbf{117}, 815-860.

\vskip.3cm\noindent Curello, G. and L. Sinander (2022) ``The Comparative Statics of Persuasion,'' \emph{working paper}.


\vskip.3cm\noindent DeMarzo, P. M., I. Kremer and A. Skrzypacz (2019) ``Test Design and Minimum Standards.'' \emph{American Economic Review}, \textbf{109}: 2173-2207. 

\vskip.3cm\noindent Doval, L., and V. Skreta. (2021) ``Constrained Information Design,'' \emph{working paper}.


\vskip.3cm\noindent Dranove, D. and G. Jin. (2010) ``Quality Disclosure and Certification: Theory and Practice.'' \emph{Journal of Economic Literature}, \textbf{48}: 935-63.

\vskip.3cm\noindent Dye, R. A. (1985) ``Disclosure of Nonproprietary Information.'' \emph{Journal of Accounting Research}, \textbf{23}: 123-145.


\vskip.3cm\noindent Eckardt, M., and S. Rathke-Doppner. (2010) ``The Quality of Insurance Intermediary Services -- Empirical Evidence for Germany." \emph{Journal of Risk and Insurance}, \textbf{77}, 667-701.

\vskip.3cm\noindent Ershov, D. and M. Mitchell. (2022) ``The Effects of Advertising Disclosure Regulations on Social Media: Evidence From Instagram.'' \textit{working paper}.





\vskip.3cm\noindent Grossman, S.  (1981)  ``The Informational Role of Warranties and Private Disclosure about Product Quality.'' \emph{Journal of Law \& Economics}, \textbf{24}: 461-484. 

\vskip.3cm\noindent Hwang, I., K. Kim, and R. Boleslavsky (2019). ``Competitive Advertising and Pricing,'' \emph{working paper}.


\vskip.3cm\noindent Inderst, R., and M. Ottaviani. (2012) ``Competition through Commissions and Kickbacks.'' \emph{American Economic Review}, \textbf{102}, 780-809.

\vskip.3cm\noindent Inderst, R., and M. Ottaviani. (2012) ``Financial Advice''. \emph{Journal of Economic Literature}, \textbf{50}, 494-512.

\vskip.3cm\noindent Inderst, R., and M. Ottaviani. (2012) ``How (not) to Pay for Advice: A Framework for Consumer Financial Protection.'' \emph{Journal of Financial Economics}, \textbf{105}, 393-411.

\vskip.3cm\noindent Jung, W.-O. and Y. K. Kwon (1988) ``Disclosure When the Market Is Unsure of Information Endowment of Managers.'' \emph{Journal of Accounting Research}, \textbf{26}: 146-153. 

\vskip.3cm\noindent Kamenica, E. (2019) ``Bayesian Persuasion and Information Design.'' \emph{Annual Review of Economics}, \textbf{11}.

\vskip.3cm\noindent Kamenica, E., and M. Gentzkow. (2011) ``Bayesian Persuasion.'' \emph{American Economic Review}, \textbf{101}, 2590-2615.

\vskip.3cm\noindent Kartik, N., F. X. Lee, and W. Suen. (2017) ``Investment in Concealable Information by Biased Experts.'' \emph{The RAND Journal of Economics}, \textbf{48}: 24-43.



\vskip.3cm\noindent Lipnowski, E. and D. Ravid. (2020) ``Cheap Talk with Transparent Motives.'' \emph{Econometrica}, \textbf{88}: 1631-1660.

\vskip.3cm\noindent Lin, X. and C. Liu. (2022) ``Credible Persuasion,'' \emph{working paper}.

\vskip.3cm\noindent Mathevet, L., D. Pearce, and E. Stacchetti. (2022) ``Reputation for a Degree of Honesty.'' \emph{working paper}.

\vskip.3cm\noindent Matthews, Steven, and Andrew Postlewaite. (1985) ``Quality Testing and Disclosure.'' \emph{RAND Journal of Economics}: 328-340.

\vskip.3cm\noindent Mensch, J. (2021) ``Monotone Persuasion.'' \emph{Games and Economic Behavior}, \textbf{130}: 521-542.

\vskip.3cm\noindent Milgrom, P.  (1981) ``Good News and Bad News: Representation Theorems and Applications.'' \emph{The Bell Journal of Economics}, \textbf{12}: 380-391. 

\vskip.3cm\noindent Milgrom, P. (2008) ``What the Seller Won't Tell You: Persuasion and Disclosure in Markets.'' \emph{Journal of Economic Perspectives}, \textbf{22}: 115?131. 

\vskip.3cm\noindent Milgrom, P. and J. Roberts. (1986) ``Relying on the Information of Interested Parties.'' \emph{RAND Journal of Economics}: 18-32.


\vskip.3cm\noindent Onuchic, P. \textcircled{r} D. Ray. (2022) ``Conveying Value via Categories,'' \emph{Theoretical Economics}, forthcoming.

\vskip.3cm\noindent Opp, C., M. Opp, and M. Harris. (2013) ``Rating Agencies in the Face of Regulation.'' \emph{Journal of Financial Economics}, \textbf{108}, 46-61.


\vskip.3cm\noindent Rayo, L., and I. Segal (2010) ``Optimal Information Disclosure." \emph{Journal of Political Economy}, \textbf{118}, 949--987.

\vskip.3cm\noindent Seidmann, Daniel J., and Eyal Winter. (1997) ``Strategic Information Transmission with Verifiable Messages.'' \emph{Econometrica}: 163-169.

\vskip.3cm\noindent Shishkin, D. (2022) ``Evidence Acquisition and Voluntary Disclosure." \emph{working paper}.


\vskip.3cm\noindent Skreta, V., and L. Veldkamp. (2009) ``Ratings Shopping and Asset Complexity: A Theory of Ratings Inflation.'' \emph{Journal of Monetary Economics}, \textbf{56}, 678-695.

\vskip.3cm\noindent Szalay, D. (2005) ``The Economics of Clear Advice and Extreme Options.'' \emph{The Review of Economic Studies}, \textbf{72}: 1173-1198.


\section*{Appendix}

\subsection*{Proof of Theorem \ref{th:1}} 
\label{pr:1}
I first complete the proof in the main text, showing that any disclosure rule that does not satisfy the described threshold structure described can be improved upon. To that end, consider $d$ and $\hat{d}$ as in (\ref{p1}) and (\ref{p2}). We prove the following two claims used in the main text.
\begin{claim}
$d$ and $\hat{d}$ produce the same overall probability of sale.
\end{claim}

\begin{proof}
\begin{align*}
\mathbb{E}[P(y,\hat{d})]&-\mathbb{E}[P(y,d)]=\\
&=\int_\mathcal{Y}\int_\mathcal{X}\left[p(x)-p(x^{ND})\right]\left[\hat{d}(x,y)-d(x,y)\right]dF_{X|y}(x)dF_Y(y)\\
&=\int_\mathcal{X}\left[p(x)-p(x^{ND})\right]\int_\mathcal{Y}\left[\hat{d}(x,y)-d(x,y)\right]dF_{Y|x}(y)dF_X(x)=0
\end{align*}
where $F_{Y|x}$ is the profitability distribution conditional on a signal realization $x$ and $F_X$ is the marginal distribution of signal realizations. The first equality uses the definition of $P(y,d)$ and the third is due to $d$ and $\hat{d}$ disclosing each realization with the same probability, as in (\ref{eq:ddhat}).
\end{proof}

\begin{claim}
$\hat{d}$ induces a larger covariance between sales and profitability than $d$.
\end{claim}
\begin{proof}
\begin{align}
&\text{Cov}\left[y,P(y,\hat{d})\right]-\text{ Cov}\left[y,P(y,\hat{d})\right]=\mathbb{E}\left[\left(P(y,\hat{d})-P(y,d)\right)\left(y-\mathbb{E}(y)\right)\right]\nonumber\\
&=\int_\mathcal{Y}\int_\mathcal{X}\left[p(x)-p(x^{ND})\right]\left[\hat{d}(x,y)-d(x,y)\right]\left[w-\mathbb{E}(y)\right]dF_{X|y}(x)dF_Y(y)\nonumber\\
&=\int_\mathcal{X}\left[p(x)-p(x^{ND})\right]\int_\mathcal{Y}\left[\hat{d}(x,y)-d(x,y)\right]\left[y-\mathbb{E}(y)\right]dF_{Y|x}(y)dF_X(x)\label{eq:cov}
\end{align}
By the definition of $\hat{d}$, for $x<x^{ND}$, $\hat{d}(x,y)-d(x,y)\geqslant 0$ when $y<\hat{y}(x)$ and $\hat{d}(x,y)-d(x,y)\leqslant 0$ when $y>\hat{y}(x)$ -- but, as given by (\ref{eq:ddhat}), the expected difference $\hat{d}(x,y)-d(x,y)$ is $0$. This, along with the fact that $y-\mathbb{E}(y)$ is increasing in $y$, implies that 
$$\int_\mathcal{Y}\left[\hat{d}(x,y)-d(x,y)\right]\left[y-\mathbb{E}(y)\right]dF_{Y|x}(y)\leqslant 0,$$
when $x\leqslant x^{ND}$. Analogously, we can show that 
$$\int_\mathcal{Y}\left[\hat{d}(x,y)-d(x,y)\right]\left[w-\mathbb{E}(y)\right]dF_{Y|x}(y)\geqslant 0,$$ 
when $x>x^{ND}$. Moreover, these inequalities are strict for a positive measure of signal realizations. These observations, along with the fact that $p$ is strictly increasing, deliver that the expression in (\ref{eq:cov}) is strictly positive. 
\end{proof}

\noindent Now we prove the following Lemma:
\begin{lemma}
\label{lem:2}
In any optimal disclosure rule, the profitability threshold satisfies (\ref{wth}).
\end{lemma}
\begin{proof}
 Using (\ref{eq:p2}), for $y\in\mathcal{Y}$ and $x\in\mathcal{X}$, we can take a derivative of the sender's value with respect to $d(x,y)$, to get
\begin{align}
\frac{\partial \Pi}{\partial d(x,y)}=&y\left(p(x)-p(x^{ND})\right)dF_{X|y}(x)dF_Y(y)\nonumber\\
&+\left(\int_\mathcal{Y}\int_\mathcal{X}\tilde{y}\left[1-d(\tilde{y},\tilde{x})\right]dF_{X|\tilde{y}}(\tilde{x})dF_Y(\tilde{y})\right)p'(x^{ND})\frac{\partial x^{ND}}{\partial d(x,y)}\nonumber
\end{align}
Now from (\ref{avg}), we get
\begin{align*}
\frac{\partial x^{ND}}{\partial d(x,y)}=\frac{\int_\mathcal{Y}\int_\mathcal{X}(\tilde{x}-x)(1-d(\tilde{y},\tilde{x}))dF_{X|\tilde{y}}(\tilde{x})dF_Y(\tilde{y})}{\left(\int_\mathcal{Y}\int_\mathcal{X}(1-d(\tilde{y},\tilde{x}))dF_{X|\tilde{y}}(\tilde{x})dF_Y(\tilde{y})\right)^2}dF_{X|y}(x)dF_Y(y)
\end{align*}
Substituting this into the previous equation, we have
\begin{align}
\frac{\partial \Pi}{\partial d(x,y)}=&\left[y\left(p(x)-p(x^{ND})\right)-y^{ND}p'(x^{ND})(x-x^{ND})\right]dF_{X|y}(x)dF_Y(y),\label{foc3}
\end{align}
where $y^{ND}$ is the average object profitability given non-disclosure. It is easy to check that, if $x<x^{ND}$, 
\begin{align*}
\frac{\partial \Pi}{\partial d(x,y)}\begin{cases}
>0\text{, if }y<y^{ND}\left[\frac{p'(x^{ND})(x^{ND}-x)}{p(x^{ND})-p(x)}\right]\\
<0\text{, if }y>y^{ND}\left[\frac{p'(x^{ND})(x^{ND}-x)}{p(x^{ND})-p(x)}\right]
\end{cases}
\end{align*}
Conversely, if $x>x^{ND}$, 
\begin{align*}
\frac{\partial \Pi}{\partial d(x,y)}\begin{cases}
>0\text{, if }y>y^{ND}\left[\frac{p'(x^{ND})(x^{ND}-x)}{p(x^{ND})-p(x)}\right]\\
<0\text{, if }y<y^{ND}\left[\frac{p'(x^{ND})(x^{ND}-x)}{p(x^{ND})-p(x)}\right]
\end{cases}
\end{align*}
Now suppose that, with positive probability, either $d(x,y)\neq 0$ when (\ref{foc3}) is negative or $d(x,y)\neq 1$ when (\ref{foc3}) is positive. Then $d$ cannot be optimal. So any optimal disclosure rule must have a threshold structure where the profitability threshold satisfies 
$$\bar{y}(x)=y^{ND}\left[\frac{p'(x^{ND})(x^{ND}-x)}{p(x^{ND})-p(x)}\right]$$
for all $x\neq x^{ND}$. Equivalently, the profitability threshold satisfies (\ref{wth}).
\end{proof}
\noindent To complete the proof of Theorem \ref{th:1}, we show that an optimal disclosure rule exists.
\begin{lemma}
An optimal disclosure rule exists.
\end{lemma}
\begin{proof}
Let $\mathcal{D}$ be the set of all disclosure rules $d:\mathcal{Y}\times\mathcal{X}\rightarrow[0,1]$. Also define the following subset of $\mathcal{D}$, the set of threshold disclosure rules: $\mathcal{T}$ is the set of disclosure rules such that $d(x,y)\in\{0,1\}\text{ and }d(x,y)=1\Leftrightarrow (x-\bar{x})(y-\hat{y}(x))> 0$ for some $\bar{x}\in\mathcal{X}$ and $\hat{y}(x)$ satisfying, for some $\bar{y}\in\mathcal{Y}$
$$\bar{y}(x)=\bar{y}\left[\frac{p'(\bar{x})(\bar{x}-x)}{p(\bar{x})-p(x)}\right].$$
Further, define the following ``discretized version'' of $\mathcal{D}$: for each $N>0$, $d\in\mathcal{D}^N$ if $d\in\mathcal{D}$ and $d(x,y)=d(y',x')$ whenever $y_{min}+\frac{n^Y}{N}(y_{max}-y_{min})\leqslant y,y'<y_{min}+\frac{n^Y+1}{N}(y_{max}-y_{min})$ and $x_{min}+\frac{n^X}{N}(x_{max}-x_{min})\leqslant x,x'<x_{min}+\frac{n^X+1}{N}(x_{max}-x_{min})$ for some $n^Y,n^X\in\{0,...,N-1\}$. Note that the disclosure rules in $\mathcal{D}^N$ are fully defined by $N^2$ numbers in $[0,1]$.

And define also a discretized version of $\mathcal{T}$: for every $N>0$, $d\in\mathcal{T}^N$ if $d\in\mathcal{D}^N$ and there is some $\bar{x}\in\mathcal{X}$ and $\bar{y}\in\mathcal{Y}$ such that (i) $d(x,y)=1$ if $(x'-\bar{x})(y'-\hat{y}(x))> 0$ for every $x'$ and $y'$ with $y_{min}+\frac{n^Y}{N}(y_{max}-y_{min})\leqslant y,y'<y_{min}+\frac{n^Y+1}{N}(y_{max}-y_{min})$ and $x_{min}+\frac{n^X}{N}(x_{max}-x_{min})\leqslant x,x'<x_{min}+\frac{n^X+1}{N}(x_{max}-x_{min})$ for some $n^Y,n^X\in\{0,...,N-1\}$; and (ii) $d(x,y)=0$ if $(x'-\bar{x})(y'-\hat{y}(x))< 0$ for every $x'$ and $y'$ with $y_{min}+\frac{n^Y}{N}(y_{max}-y_{min})\leqslant y,y'<y_{min}+\frac{n^Y+1}{N}(y_{max}-y_{min})$ and $x_{min}+\frac{n^X}{N}(x_{max}-x_{min})\leqslant x,x'<x_{min}+\frac{n^X+1}{N}(x_{max}-x_{min})$ for some $n^Y,n^X\in\{0,...,N-1\}$.
\vskip.3cm
\noindent Now we use the following facts about these sets: 
\vskip.3cm

\noindent \textbf{Fact 1.} For every $N$, $\sup_{d\in\mathcal{D}^N}\Pi(d)=\sup_{d\in\mathcal{T}^N}\Pi(d)$.
\begin{proof}[Proof of Fact 1]
As previously noted, the disclosure rules in $\mathcal{D}^N$ are fully defined by $N^2$ numbers in $[0,1]$. This, along with the continuity of the objective, implies that there exists some $\hat{d}\in\mathcal{D}^N$ such that $\Pi(\hat{d})=\sup_{d\in\mathcal{D}^N}\Pi(d)$. But, with an argument analogous to the proof of Lemma \ref{lem:2}, we can show that if $d\neq\mathcal{T}^N$, then $\Pi(d)<\sup_{d\in\mathcal{D}^N}\Pi(d)$. And so it must be that $\hat{d}\in\mathcal{T}^N$, completing the proof.
\end{proof}

\noindent \textbf{Fact 2.} There exists some $d^*$ such that $\Pi(d^*)=\sup_{d\in\mathcal{T}}\Pi(d)$.
\begin{proof}[Proof of Fact 2]
Disclosure rules in $\mathcal{T}$ are fully defined by two numbers, $\bar{x}\in\mathcal{X}$ and $\bar{y}\in\mathcal{Y}$, whenever $(x-\bar{x})(y-\hat{y}(x))\neq0$. But, because the joint distribution of profitabilities and signal realizations have no mass points, the disclosure rule at the equality points do not affect the seller's expected payoff. And so, effectively, when choosing a disclosure rule in $\mathcal{T}$, the seller chooses only two numbers, from a compact set. This, along with continuity of the objective, implies the existence of some $d^*$ such that $\Pi(d^*)=\sup_{d\in\mathcal{T}}\Pi(d)$.
\end{proof}

Now we have $\lim_{N\rightarrow\infty}\sup_{\mathcal{D}^N}\Pi(d)=\sup_{\mathcal{D}}\Pi(d)$ and $\lim_{N\rightarrow\infty}\sup_{\mathcal{T}^N}\Pi(d)=\sup_{\mathcal{T}}\Pi(d)$. But, by Fact 1, $\sup_{d\in\mathcal{D}^N}\Pi(d)=\sup_{d\in\mathcal{T}^N}\Pi(d)$ for every $N$. And so $\sup_{d\in\mathcal{D}}\Pi(d)=\sup_{d\in\mathcal{T}}\Pi(d)$. But then, by Fact 2, there exists some $d^*$ such that $\Pi(d^*)=\sup_{d\in\mathcal{T}}\Pi(d)=\sup_{d\in\mathcal{D}}\Pi(d)$, which concludes the existence proof.
\end{proof}

\subsection*{Proof of Corollary \ref{prop1}}
If a disclosure rule conceals signal realizations with strictly positive probability, then it must be that $x^{ND}\in\text{int}\left(\mathcal{X}\right)$ and $y^{ND}\in\text{int}\left(\mathcal{Y}\right)$. Next, I argue that the fully revealing disclosure rule is not optimal.

Remember that the demand function is affine, so $p(x)=a+bx$. Then it must be that for any disclosure rule $d$, $\mathbb{E}[P(y,d)]=a\mathbb{E}(y)+b\mathbb{E}(x)$. (This is a straightforward consequence of the martingale property of beliefs, or ``Bayesian plausibility.'') 
Therefore the overall probability of sale is independent of the disclosure rule.

Now consider a full revelation disclosure rule $d^1$ and some other threshold disclosure rule $d$, defined by interior thresholds $\bar{x}$ and $\bar{y}$. It must be that, for $y<\bar{y}$, $P(y,d)<P(y,d^1)$ and, for $y>\bar{y}$, $P(y,d)>P(y,d^1)$, and so the covariance of sales and profitability is larger under $d$ than under $d^1$; and thus $d^1$ is not an optimal disclosure rule. (This last argument relies on the fact that the distribution of profitabilities and signal realizations has full support.)\qed

\subsection*{Proof of Lemma \ref{lem:1}}
Suppose the seller chooses disclosure rule $d$. Upon observing non-disclosure and profitability $y$ (under mandated transparency), the buyer's mean posterior is
$$x_y^{ND}=\frac{\int_\mathcal{X}[1-d(x,y)] xdF_{X|y}(x)}{\int_\mathcal{X}[1-d(x,y)] dF_{X|y}(x)}.$$
And so the probability of sale of object $y$ is
$$P(y,d)=\int_\mathcal{X}\left[d(x,y)p(x)+(1-d(x,y))p(x_y^{ND})\right]dF_{X|y}(x).$$ 
Note that the probability of sale of object $y$ is thus independent of the disclosure rule used for objects with profitability $y'\neq y$. And so the seller's problem is separable across profitability levels. Therefore maximizing $\Pi(d)$ over $d:\mathcal{Y}\times\mathcal{X}\rightarrow[0,1]$ is equivalent to maximizing, for each $y\in\mathcal{Y}$, $yP(y,d(\cdot,y))$, over $d(\cdot,y):\mathcal{X}\rightarrow[0,1]$.\qed

\subsection*{Proof of Proposition \ref{propp1}}
For any given $d(\cdot,y):\mathcal{X}\rightarrow[0,1]$, 
$$yP(y,d(\cdot,y))=\int_\mathcal{X}yp(x)dF^B_y(x,d(\cdot,y)),$$
where $F^B_y$ is the distribution of posterior means induced on the buyer when profitability is $y$ and the disclosure rule is $d(\cdot,y)$.
Suppose $p$ is everywhere strictly concave, and take two disclosure rules $d(\cdot,y)$ and $d'(\cdot,y)$, where $F^B_y(\cdot,d(\cdot,y))$ is a mean preserving spread of $F^B_y(\cdot,d'(\cdot,y))$. 

This immediately implies that, $\mathbb{E}\left[P(d(\cdot,y))\right]\leqslant \mathbb{E}\left[P(d'(\cdot,y))\right]$. Now let $d^0:\mathcal{X}\rightarrow[0,1]$ be the fully concealing disclosure rule -- that is, for all $x\in\mathcal{X}$, $d^0(x)=0$. And note that, for any $d(\cdot,y)$, $F^B_y(\cdot,d(\cdot,y))$ is a mean preserving spread of $F^B_y(\cdot,d^0)$, and so $d^0$ maximizes $\mathbb{E}\left[yP(y,d(\cdot,y))\right]$. By Lemma \ref{lem:1}, it is then a solution to the seller's problem.

An analogous argument shows that, if $p$ is everywhere strictly convex, then the fully disclosing rule $d^1$ maximizes $\mathbb{E}\left[yP(y,d(\cdot,y))\right]$, and $d^1$ is thus a solution to the seller's problem.\qed

\subsection*{Proof of Proposition \ref{pr:p3}}
When the buyer observes the object's profitability with probability $\tau\in[0,1]$, the seller's profit, as a function of the chosen disclosure rule $d$, is given by
\begin{align}
\Pi_\tau(d)=\int_{\mathcal{Y}}y\left[\tau P^1(y,d)+(1-\tau)P^0(y,d)\right]dF_Y(y),
\label{eq:l1}
\end{align}
$$\text{where }P^1(y,d)=\int_\mathcal{X}\left[d(x,y)p(x)+(1-d(x,y))p(x_y^{ND}))\right]dF_{X|y}(x),$$ 
$$\text{and }P^0(y,d)=\int_\mathcal{X}\left[d(x,y)p(x)+(1-d(x,y))p(x^{ND}))\right]dF_{X|y}(x).$$ 
Remember that $x^{ND}$ is as defined in (\ref{avg}) and $x_y^{ND}$ is defined in (\ref{eqq1}).

Suppose $p$ is everywhere strictly convex, and assume by contradiction that $d$ has strictly less disclosure than $d'$. A consequence is that, conditional on the profitability $y$ being disclosed to the observer, $d$ implies that the buyer is strictly less informed about the object's value than $d'$. And therefore 
\begin{equation}
\label{e1}\int_{\mathcal{Y}} P^1(y,d)dF_Y(y)<\int_{\mathcal{Y}} P^1(y,d')dF_Y(y).\end{equation}
 Consequently, it must be that 
 \begin{equation}
 \label{e2}\int_{\mathcal{Y}} P^0(y,d)dF_Y(y)>\int_{\mathcal{Y}} P^0(y,d')dF_Y(y),\end{equation} for otherwise $d$ cannot be the optimal policy under transparency $\tau$. Now using (\ref{e1}) and (\ref{e2}), and the fact that $\tau\geqslant\tau'$, it must be that if $\Pi_\tau(d)>\Pi_\tau(d')$, then $\Pi_{\tau'}(d)>\Pi_{\tau'}(d')$, which contradicts the optimality of $d'$ under policy $\tau'$. And thus we conclude that $d$ must have no less disclosure than $d'$. (An analogous argument implies that if $p$ is everywhere strictly concave, then $d$ must have no more disclosure than $d'$.)
 
Now suppose $1\geqslant\tau>\tau'=0$.
From (\ref{eq:l1}), we have:
\begin{align}
&\frac{\partial \Pi_\tau}{\partial d(x,y)}=\tau\left[y\left(p(x)-p(x_y^{ND})\right)-yp'(x_y^{ND})(x-x_y^{ND})\right]dF_{X|y}(x)dF_Y(y)\nonumber\\
&\hskip5pt+(1-\tau)\left[y\left(p(x)-p(x^{ND})\right)-y^{ND}p'(x^{ND})(x-x^{ND})\right]dF_{X|y}(x)dF_Y(y).\label{foc5}
\end{align}
Let $d_0\in\argmax_d\Pi_0(d)$ and $d_\tau\in\argmax_d\Pi_\tau(d)$ for some $\tau\in(0,1)$. Then it must be that almost surely $d_0(x,y),d_\tau(x,y)\in\{0,1\}$ and $d_0(x,y)=0$ ($d_\tau(x,y)=0$) if and only if $\frac{\partial \Pi_0(d_0)}{\partial d(x,y)}\leqslant0$ ($\frac{\partial \Pi_\tau(d_\tau)}{\partial d(x,y)}\leqslant0$).

Let $ND_0=\left\{(x,y)\in\mathcal{Y}\times\mathcal{X}:\frac{\partial \Pi_0(d_0)}{\partial d(x,y)}\leqslant0\right\}$ be the non-disclosure set of $d_0$, 
 and analogously define $ND_\tau=\left\{(x,y)\in\mathcal{Y}\times\mathcal{X}:\frac{\partial \Pi_\tau(d_\tau)}{\partial d(x,y)}\leqslant0\right\}$, the non-disclosure set of $d_\tau$.

Now suppose the demand function $p$ is strictly concave. Then the first term in (\ref{foc5}) is negative, for all $(x,y)$. And so, if $x^{ND}(d_0)=x^{ND}(d_\tau)$ and $y^{ND}(d_0)=y^{ND}(d_\tau)$, it must be that $ND_0\subset ND_\tau$. And consequently, disclosure rule $d_0$ discloses more signal realizations than disclosure rule $d_\tau$; and $d_0$ has strictly more disclosure than $d_\tau$. If instead $p$ is strictly convex, then  the first term in (\ref{foc5}) is positive, for all $(x,y)$. And the opposite conclusion holds: if $x^{ND}(d_0)=x^{ND}(d_\tau)$ and $y^{ND}(d_0)=y^{ND}(d_\tau)$, $ND_\tau\subset ND_0$; and therefore $d_0$ has strictly less disclosure than $d_\tau$. 
\qed

\subsection*{Proof of Proposition \ref{propp2}}
Suppose $p(x)=a+bx$, with $a,b\in\mathbb{R}$. 
Then, under full transparency, for any disclosure rule $d$,
$$yP(y,d)=\int_\mathcal{X}y[a+bx]dF^B_y(x,d_y),$$
where $d_y(\cdot)=d(y,\cdot)$ and $F^B_y$ is the distribution of posterior means induced on the buyer when profitability is $y$ and the disclosure rule is $d_y$. And so 
$$yP(y,d)=ya+yb\int_\mathcal{X}xdF^B_y(x,d_y)=ya+yb\mathbb{E}[x|y].$$
The second equality is due to the fact that, under any disclosure rule, $F^B_y$ has average equal to $\mathbb{E}[x|y]$, the underlying expected value of the object given that the profitability is $y$. ($F^B_y$ is a mean preserving spread of the degenerate distribution at the object's expected value, given the profitability that is observed by the buyer due to the mandated transparency.) Therefore, under full transparency, we have
\begin{equation}
\Pi^1(d)=\mathbb{E}\left[yP(y,d)\right]=a\mathbb{E}(y)+b\mathbb{E}\left[y\mathbb{E}(x|y)\right],\label{eq:j1}
\end{equation}
which is independent of the disclosure rule, proving part (a) of the proposition. (Denote by $\Pi^1$ the seller's value under full transparency; and by $\Pi^0$ the value when profitability is not revealed to the buyer.)

Now if $\tau\neq1$, then the seller's value is
$$
\Pi(d)=\tau\Pi^1(d)+(1-\tau)\Pi^0(d)=\tau\left\{a\mathbb{E}(y)+b\mathbb{E}\left[y\mathbb{E}(x|y)\right]\right\}+(1-\tau)\Pi^0(d).
$$
Because the first term is constant, it follows that the seller's optimal disclosure rule maximizes $\Pi^0(d)$, proving part (b) of the proposition.\qed

\subsection*{Proof of Proposition \ref{propn}}
Let $\Pi_\tau(d,\theta)$ be the expected payoff to the seller who has the signal of precision $\theta$, uses disclosure rule $d$, and has their motives disclosed to the buyer with probability $\tau$. And let $\bar{\Pi}_\tau(\theta)=\max_d\Pi_\tau(d,\theta)$.
\begin{lemma}
$\bar{\Pi}_0(\theta)$ is weakly increasing in $\theta$.
\label{lem:j1}
\end{lemma}
\begin{proof}[Proof of Lemma \ref{lem:j1}]
Remember that $p(x)=a+bx$. Then for each $\theta$, and each $y$,
$$yP(y,d,\theta)=ay+by\left[\int_\mathcal{X}xd(x,y)dF_{X|y}(x,\theta)+\int_\mathcal{X}x^{ND}(1-d(x,y))dF_{X|y}(x,\theta)\right],$$
\begin{equation*}
\text{where }x^{ND}=\frac{\int_\mathcal{Y}\int_\mathcal{X}x\left(1-d(x,y)\right)dF_{X|y}(x,\theta)dY(y)}{\int_\mathcal{Y}\int_\mathcal{X}\left(1-d(x,y)\right)dF_{X|y}(x,\theta)dY(y)}.
\end{equation*}
Now take two precision levels $\theta'>\theta$, and remember that, for each $y$, $F_{X|y}(\cdot,\theta')$ is a mean-preserving spread of $F_{X|y}(\cdot,\theta)$. And so, for any disclosure rule $d$, there exists some disclosure rule $d'$ such that
$$\int_\mathcal{X}xd'(x,y)dF_{X|y}(x,\theta')=\int_\mathcal{X}xd(x,y)dF_{X|y}(x,\theta),$$
$$\text{and }\int_\mathcal{X}d'(x,y)dF_{X|y}(x,\theta')=\int_\mathcal{X}d(x,y)dF_{X|y}(x,\theta),$$
for every $y\in\mathcal{Y}$. And thus $yP(y,d',\theta')=yP(y,d,\theta)$, and $\Pi_0(d',\theta')=\Pi_0(d,\theta)$. Because this is true for any disclosure rule $d$, it must be that $\bar{\Pi}_0(\theta')\geqslant \bar{\Pi}_0(\theta)$.
\end{proof}
\begin{lemma}
\label{lem:j2}
$\bar{\Pi}_1(\theta)$ is independent of $\theta$.
\end{lemma}
\begin{proof}
From equation (\ref{eq:j1}), we know that under full transparency, 
\begin{equation}\bar{\Pi}_1(\theta)=a\mathbb{E}(y)+b\mathbb{E}\left[y\mathbb{E}(x|y,\theta)\right]=a\mathbb{E}(y)+b\mathbb{E}\left[y\mathbb{E}(x|y)\right].
\label{eq:j2}
\end{equation}
The last equality is due to the fact that the signal's precision does not affect the underlying expected value of the object, for a given $y$.
\end{proof}
\noindent By Lemmas \ref{lem:j1} and \ref{lem:j2},  and the fact that $\Pi(d,\theta)=\tau\Pi^1(d,\theta)+(1-\tau)\Pi^0(d,\theta)$ for every disclosure rule $d$, we have
$$\bar{\Pi}_\tau(\theta)=\tau\bar{\Pi}_1(\theta)+(1-\tau)\bar{\Pi}_0(\theta),$$
where the first term is constant and the second term is weakly increasing in $\theta$.

Because $\bar{\Pi}_1(\theta)$ is constant and $c$ is strictly increasing in $\theta$, it must be that if $\tau=1$, the seller acquires $\theta=0$ (proving the first statement in the proposition).
Now take two transparency levels $\tau'>\tau$ and suppose $\theta$ and $\theta'$ are signal precisions optimally acquired by the seller under each transparency level. If $\theta=\theta'$, then the second statement in the proposition is true. Suppose instead that $\theta\neq\theta'$. Then, by the optimality of $\theta$ under $\tau$,
$$\tau\bar{\Pi}_1(\theta)+(1-\tau)\bar{\Pi}_0(\theta)-c(\theta)\geqslant \tau\bar{\Pi}_1(\theta')+(1-\tau)\bar{\Pi}_0(\theta')-c(\theta').$$
\begin{equation}
\Rightarrow (1-\tau)\bar{\Pi}_0(\theta)-c(\theta)\geqslant (1-\tau)\bar{\Pi}_0(\theta')-c(\theta').
\label{eq:jj1}
\end{equation}
And by the optimality of $\theta'$ under $\tau'$,
$$\tau'\bar{\Pi}_1(\theta')+(1-\tau')\bar{\Pi}_0(\theta')-c(\theta')\geqslant \tau'\bar{\Pi}_1(\theta)+(1-\tau')\bar{\Pi}_0(\theta)-c(\theta).$$
\begin{equation}
\Rightarrow(1-\tau')\bar{\Pi}_0(\theta')-c(\theta')\geqslant (1-\tau')\bar{\Pi}_0(\theta)-c(\theta).
\label{eq:jj2}
\end{equation}
Adding up the two inequalities, we have
$$(\tau'-\tau)\bar{\Pi}_0(\theta)\geqslant (\tau'-\tau)\bar{\Pi}_0(\theta')\Rightarrow \bar{\Pi}_0(\theta)\geqslant \bar{\Pi}_0(\theta').$$
Note that $\bar{\Pi}_1(\theta)=\bar{\Pi}_1(\theta')$ contradicts either (\ref{eq:jj1}) or (\ref{eq:jj2}), because $c$ is strictly increasing in $\theta$. And so it must be that  $\bar{\Pi}_0(\theta)>\bar{\Pi}_0(\theta')$, and so $\theta>\theta'$, completing the proof.\qed

\subsection*{Proof of Proposition \ref{propnn}}
Statement 2 in the proposition follows trivially from the fact that, under $\tau=1$, the seller acquires $\theta=0$, and the signal  $\theta=0$ is perfectly uninformative. 
\vskip.3cm
\noindent\textbf{Statement 1.} $F^B(\cdot,\delta)$ is not strictly Blackwell more informative than $F^B(\cdot,\delta')$.
\begin{proof}[Proof of Statement 1]
If $\tau=1$, then statement 1 is a consequence of statement 2. Assume instead that $\tau\neq1$.
Because $d$ is an optimal disclosure rule under $\tau$, it is a threshold disclosure rule (by Theorem \ref{th:1}), and so $F^B(\cdot,\delta)$ is defined as follows. There is some $\bar{x}\in\text{int}\left(\mathcal{X}\right)$ and $\bar{y}\in\text{int}\left(\mathcal{Y}\right)$ such that:

\noindent$\circ$ If $x<\bar{x}$, 
\begin{equation*}F^B(x,\theta,d)=\int_{[y_{min},\bar{y})}F_{X|y}(x,\theta)dF_Y(y).\end{equation*}
$\circ$ If $x\geqslant \bar{x}$,
\begin{align*}
F^B(x,\theta,d)&=\int_\mathcal{Y}F_{X|y}(x,\theta)dF_Y(y)+\int_{[y_{min},\bar{y})}(1-F_{X|y}(x,\theta))dF_Y(y)\\
&=F_Y(\bar{y})+\int_{[\bar{y},y_{max}]}F_{X|y}(x,\theta)dF_Y(y).\end{align*} 
The distribution $F^B(\cdot,\delta')$ can be expressed analogously, substituting $\theta$ with $\theta'$ and $\bar{x}$ and $\bar{y}$ with some potentially different $\bar{x}'\in\text{int}\left(\mathcal{X}\right)$ and $\bar{y}'\in\text{int}\left(\mathcal{Y}\right)$. Consider three cases:\footnote{The following arguments repeatedly use the definition of mean preserving spreads in Shaked and Shanthikumar (2007), equation (3.A.8).}

\underline{Case 1.} $\bar{x}=\bar{x}'$ and $\bar{y}=\bar{y}'$. In this case, the fact that for each $y$, $F_{X|y}(\cdot,\theta')$ is a mean-preserving spread of $F_{X|y}(\cdot,\theta)$ implies that $F^B(x,\delta')$ is a mean preserving spread of $F^B(x,\delta)$. (And so the former is Blackwell more informative.)

\underline{Case 2.} $\bar{y}\neq\bar{y}'$. Suppose first that $\bar{y}'>\bar{y}$. Then there is some $x<\min\{\bar{x},\bar{x}'\}$ such that
$$\int_{x_{min}}^{x}F^B(\hat{x},\delta')d\hat{x}=\int_{x_{min}}^{x}\int_{[y_{min},\bar{y}')}F_{X|y}(\hat{x},\theta')dF_Y(y)d\hat{x}$$
$$=\int_{y_{min}}^{\bar{y}'}\int_{x_{min}}^{x}F_{X|y}(\hat{x},\theta')d\hat{x}dF_Y(y)>\int_{y_{min}}^{\bar{y}}\int_{x_{min}}^{x}F_{X|y}(\hat{x},\theta')d\hat{x}dF_Y(y)$$
\begin{equation}\geqslant\int_{y_{min}}^{\bar{y}}\int_{x_{min}}^{x}F_{X|y}(\hat{x},\theta)d\hat{x}dF_Y(y)=\int_{x_{min}}^{x}F^B(\hat{x},\delta)d\hat{x}.
\label{eq:j3}
\end{equation}
The first inequality used the fact that $\bar{y}'>\bar{y}$ and that for each $y$, $F_{X|y}(\cdot,\theta')$ has full support over $\mathcal{X}(\theta')$ (and thus over $\mathcal{X}(\theta)\subseteq\mathcal{X}(\theta')$). The second inequality used the fact that for each $y$, $F_{X|y}(\cdot,\theta')$ is a mean-preserving spread of $F_{X|y}(\cdot,\theta)$.
But (\ref{eq:j3}) implies that $F^B(\hat{x},\delta)$ is not a mean preserving spread of $F^B(\hat{x},\delta')$.

An analogous argument can be used to show that, if $\bar{y}'<\bar{y}$, there is some $x>\max\{\bar{x},\bar{x}'\}$ such that $\int_{[x,x_{max})}(1-F^B(\hat{x},\delta'))d\hat{x}>\int_{[x,x_{max})}(1-F^B(\hat{x},\delta))d\hat{x}$, and so $F^B(\hat{x},delta)$ is not a mean preserving spread of $F^B(\hat{x},\delta')$.\footnote{This argument would use the equivalent definition of mean preserving spreads, in equation (3.A.7) of Shaked and Shanthikumar (2007).}

\underline{Case 3.} $\bar{y}\neq\bar{y}'$ and $\bar{x}\neq\bar{x}'$. Suppose first that $\bar{x}'<\bar{x}$. Then there exists some $x\in(\bar{x}',\bar{x})$ such that
$$\int_{x_{min}}^{x}F^B(\hat{x},\delta')d\hat{x}=\int_{x_{min}}^{x}\left[Y(\bar{y})+\int_{\bar{y}}^{y_{max}}F_{X|y}(\hat{x},\theta)dF_Y(y)\right]d\hat{x}$$
$$>\int_{x_{min}}^{x}\int_{y_{min}}^{\bar{y}}F_{X|y}(\hat{x},\theta')dF_Y(y)d\hat{x}=\int_{y_{min}}^{\bar{y}}\int_{x_{min}}^{x}F_{X|y}(\hat{x},\theta')d\hat{x}dF_Y(y)$$
$$\geqslant \int_{y_{min}}^{\bar{y}}\int_{x_{min}}^{x}F_{X|y}(\hat{x},\theta)d\hat{x}dF_Y(y)=\int_{x_{min}}^{x}F^B(\hat{x},\delta)d\hat{x}.$$
The first inequality used the fact that $\bar{y}'>\bar{y}$ and that for each $y$, $F_{X|y}(\cdot,\theta')$ has full support over $\mathcal{X}(\theta')$ (and thus over $\mathcal{X}(\theta)\subseteq\mathcal{X}(\theta')$). The second inequality used the fact that for each $y$, $F_{X|y}(\cdot,\theta')$ is a mean-preserving spread of $F_{X|y}(\cdot,\theta)$.
Again, this implies that $F^B(\hat{x},\delta)$ is not a mean preserving spread of $F^B(\hat{x},\delta')$. An analogous argument can be used to show that if $\bar{x}'>\bar{x}$, there is some $x\in(\bar{x},\bar{x}')$ such that $\int_{[x,x_{max})}(1-F^B(\hat{x},\delta'))d\hat{x}>\int_{[x,x_{max})}(1-F^B(\hat{x},\delta))d\hat{x}$ and so $F^B(\hat{x},\delta)$ is not a mean preserving spread of $F^B(\hat{x},\delta')$.
\end{proof}

\noindent\textbf{Statement 3.} If signals of precision $\theta$ and $\theta'$ are regular and symmetric, then $F^B(\cdot,\delta')$ is Blackwell more informative than $F^B(\cdot,\delta)$.
\begin{proof}[Proof of Statement 3]
If a signal of precision $\theta$ is regular, then if 
\begin{equation}
\mathbb{E}_\theta\left[x|(x-\hat{x})(y-\hat{y})<0\right]=\hat{x}\text{ and }\mathbb{E}_\theta\left[y|(x-\hat{x})(y-\hat{y})<0\right]=\hat{y},
\label{eq:j4}
\end{equation}
for some $\hat{x}\in\text{int}\left(\mathcal{X}\right)$ and $\hat{y}\in\text{int}\left(\mathcal{Y}\right)$, then the threshold disclosure rule with thresholds $\hat{x}$ and $\hat{y}$ is the unique optimal disclosure rule. (Corollary \ref{prop1} implies that any optimal disclosure rule satisfies (\ref{eq:j4}) and regularity implies that there is a unique such pair $(\hat{x},\hat{y})$.)

Now suppose the signal of precision $\theta$ is symmetric and note that
\begin{equation}\mathbb{E}_\theta\left[x|(x-\mathbb{E}(x))(y-\mathbb{E}(y))<0\right]=
\label{eq:j5}
\end{equation}
$$=\int_{y_{min}}^{\mathbb{E}(y)}\int_{\mathbb{E}(x)}^{x_{max}}xdF_{X|y}(x,\theta)dF_Y(y)+\int_{\mathbb{E}(y)}^{y_{max}}\int_{x_{min}}^{\mathbb{E}(x)}xdF_{X|y}(x,\theta)dF_Y(y)$$
$$=\int_{y_{min}}^{\mathbb{E}(y)}\int_{\mathbb{E}(x)}^{x_{max}}xdF_{X|y}(x,\theta)dF_Y(y)+\int_{y_{min}}^{\mathbb{E}(y)}\int_{x_{min}}^{\mathbb{E}(x)}xdF_{X|y}(x,\theta)dF_Y(y)$$
$$=\int_{y_{min}}^{\mathbb{E}(y)}\int_{x_{min}}^{x_{max}}xdF_{X|y}(x,\theta)dF_Y(y)=\mathbb{E}(x),$$
where $\mathbb{E}(x)$ and $\mathbb{E}(y)$ are the underlying unconditional expected value and profitability of the object, respectively.
The second and last equalities in (\ref{eq:j5}) use the fact that the signal is symmetric. Analogously, we can show that 
$$\mathbb{E}_\theta\left[y|(x-\mathbb{E}(x))(y-\mathbb{E}(y))<0\right]=\mathbb{E}(y).$$
And so the optimal disclosure rule given acquired signal $\theta$ must use thresholds $\bar{x}=\mathbb{E}(x)$ and $\bar{y}=\mathbb{E}(y)$. But the same is true for $\theta'$, because signal $\theta'$ is also regular and symmetric. And so, by Case 1 in the proof of Statement 1, it must be that $F^B(x,\delta')$ is a mean preserving spread of $F^B(x,\delta)$. (And so the former is Blackwell more informative.)
\end{proof}

\subsection*{Proof of Claim in Example \ref{sec:ex}}
In Example \ref{sec:ex}, I made the following claim:
\begin{claim}
The signal structure described by $Y=U[0,1]$ and $F_y(\cdot;\theta)=U\left[\frac{1-\theta}{2},\frac{1+\theta}{2}\right]$ for every $y$ is regular. 
 \end{claim}
\begin{proof}
For cleaner algebra, I write the proof for the case with $\theta=1$, so that $F_y(\cdot;\theta)=U\left[0,1\right]$. The proof for $\theta\in[0,1)$ follows the same steps.

Remember that a signal of precision $\theta$ is \emph{regular} if there is a unique pair $(\hat{x},\hat{y})\in\text{int}\left(\mathcal{X}\times\mathcal{Y}\right)$ such that
\begin{equation}
\mathbb{E}_\theta\left[x|(x-\hat{x})(y-\hat{y})<0\right]=\hat{x}\text{ and }\mathbb{E}_\theta\left[y|(x-\hat{x})(y-\hat{y})<0\right]=\hat{y}.
\label{eq:jp1}
\end{equation}
Given the distributions in the example, (\ref{eq:jp1}) becomes
$$\frac{\hat{y}\hat{x}\frac{\hat{x}}{2}+(1-\hat{y})(1-\hat{x})\frac{1+\hat{x}}{2}}{\hat{x}\hat{y}+(1-\hat{y})(1-\hat{x})}=\hat{x}\text{ and }\frac{\hat{y}\hat{x}\frac{\hat{y}}{2}+(1-\hat{y})(1-\hat{x})\frac{1+\hat{y}}{2}}{\hat{x}\hat{y}+(1-\hat{y})(1-\hat{x})}=\hat{y}.$$
$$\Rightarrow\hat{y}\hat{x}\hat{x}=(1-\hat{y})(1-\hat{x})(1-\hat{x})\text{ and }\hat{y}\hat{y}\hat{x}=(1-\hat{y})(1-\hat{y})(1-\hat{x}).$$
Remember that we are looking for an interior solution, so $\hat{x},\hat{y}\in(0,1)$. Dividing the first equation by the second, we have
$$\frac{\hat{x}}{\hat{y}}=\frac{1-\hat{x}}{1-\hat{y}},$$
which has a unique solution at $\hat{x}=\hat{y}=1/2$.
\end{proof}

\subsection*{Proof of Proposition \ref{pr:5}}
Assume $y_{min}<0<y_{max}$. (The case where all profitabilities are positive or all profitabilities are negative is discussed in the main text.)

Suppose the buyer's mean posterior about the object's value after observing non-disclosure is $\hat{x}\in\mathcal{X}$. 
Then a seller with profitability $y>0$, upon observing a signal realization $x$, will choose to disclose it if and only if $x\geqslant\hat{x}$ (assuming that the seller discloses when indifferent). Conversely, a seller with profitability $y<0$ will disclose signal realization $x$ if and only if $x\leqslant \hat{x}$. A seller with profitability $y=0$ is always indifferent between disclosing and not disclosing, and we resolve this indifference with disclosure. 

And so Bayesian consistency requires that, in this equilibrium, 
$$\hat{x}=\mathbb{E}\left[x|(x-\hat{x})y<0\right],$$
where we note that the set $\{(x-\bar{x})y<0\}$ has positive probability, because the joint distribution of profitabilities and signal realizations has full support.

Regarding existence, I remark that there exists some $\hat{x}$ satisfying 
$$\hat{x}=\mathbb{E}\left[x|(x-\hat{x})y<0\right].$$
To see that, note that $\mathbb{E}\left[x|(x-\hat{x})y<0\right]\in(x_{min},x_{max})$ for both $\hat{x}=x_{min}$ and $\hat{x}=x_{max}$ (both because the joint distribution of profitabilities and signal realizations has full support). And so, by continuity of $\mathbb{E}\left[x|(x-\hat{x})y<0\right]$, there must be some solution to $\hat{x}=\mathbb{E}\left[x|(x-\hat{x})y<0\right]$. Given such $\hat{x}$, it is easy to see that $d^*(x,y)\in\{0,1\}$ and $d^*(x,y)=1\Leftrightarrow (x-\hat{x})y\geqslant 0$ defines an equilibrium disclosure strategy.\qed

\subsection*{Proof of Proposition \ref{pr:6}}
Under mandated transparency, the buyer forms mean posteriors about the object separately for each profitability level $y$. And so it follows from standard unravelling arguments that there can be no concealment in equilibrium.\qed

\end{document}